\documentclass[journal]{IEEEtran}
\usepackage{amsmath,amsfonts}
\usepackage{algorithmic}
\usepackage{algorithm}
\usepackage{array}
\usepackage{graphicx}
\usepackage[caption=false,font=normalsize,labelfont=sf,textfont=sf]{subfig}
\usepackage{amsmath,amsfonts}
\usepackage{array}
\usepackage{pdfpages}
\usepackage{mathtools, cuted}
\usepackage{comment}
\usepackage{tcolorbox}

\newtcolorbox{myblock}{
  colback=blue!10,
  colframe=blue!50!black,
  fonttitle=\bfseries
}

\usepackage{comment}
\usepackage{multirow}
\usepackage{amsthm,amssymb}
\newtheorem{theorem}{Theorem}[section]
\newtheorem{proposition}[theorem]{Proposition}
\newtheorem{lemma}[theorem]{Lemma}

\theoremstyle{definition}
\newtheorem{definition}[theorem]{Definition}

\theoremstyle{remark}
\newtheorem{remark}[theorem]{Remark}

\newtheorem{problem}{Problem}

\usepackage{graphicx}
\usepackage{makecell}
\usepackage{booktabs}
\usepackage{algorithm}
\usepackage{algorithmic}
\usepackage{url}
\usepackage{textcomp}
\usepackage{stfloats}
\usepackage{url}
\usepackage{verbatim}
\usepackage{graphicx}
\usepackage{cite}
\begin{document}

\title{Age of Information Diffusion on Social Networks}
\author{Songhua Li and Lingjie Duan, \IEEEmembership{Senior Member, IEEE}
\thanks{A preliminary version of this work appears 
at MobiHoc'23 \cite{li2023age}. 
Songhua Li and Lingjie Duan are with the Pillar of Engineering Systems and Design, Singapore University of Technology and Design, Singapore.  \\Email: \{songhua\_li,lingjie\_duan\}@sutd.edu.sg.}
}

 \markboth{}%
{Songhua Li and Lingjie Duan: Age of Information Diffusion in Social Networks}


\maketitle

\begin{abstract}
To promote viral marketing, major
social platforms (e.g., Facebook Marketplace and Pinduoduo) repeatedly select and invite different users (as seeds) in online social networks to share fresh information about a product or service with their friends. Thereby, we are motivated to optimize a multi-stage seeding process of viral marketing in social networks, and adopt the recent notions of the peak and the average age of information (AoI) to measure the timeliness of promotion information received by network users.  Our problem is different from the literature on information diffusion in social networks, which limits to one-time seeding and overlooks AoI dynamics or information replacement over time.  
As a critical step, we manage to develop closed-form expressions that characterize and trace AoI dynamics over any social network.
For the peak AoI problem, we first prove the NP-hardness of our multi-stage seeding problem by a highly non-straightforward reduction from the dominating set problem, and
then
present a new polynomial-time algorithm that achieves good approximation guarantees (e.g., less than 2 for linear network topology).  To minimize the average AoI, we also prove that our problem is NP-hard by properly reducing it from the set cover problem. Benefiting from our two-sided bound analysis on the average AoI objective, we build up a new framework for approximation analysis and link our problem to a much simplified sum-distance minimization problem. This intriguing connection inspires us to develop another polynomial-time algorithm that achieves a good approximation guarantee. Additionally, our theoretical results are well corroborated by experiments on a real social network.
\end{abstract}
\begin{IEEEkeywords}
Age of information, social  network, multi-stage seeding, NP-hardness, approximation algorithms
\end{IEEEkeywords}
\section{Introduction}
\IEEEPARstart{T}{oday}, omnipresent online social networks (e.g., Facebook and WeChat) have revolutionized the way that people interact and share information, creating viral marketing opportunities for social commerce platforms (e.g., Facebook Marketplace and Pinduoduo) \cite{chang2020elaboration,chopra2023systematic}. Pinduoduo, for example, promotes its products by selecting and inviting its users to push and share promotion information via their WeChat accounts \cite{Pinduoduo}, and these users are well motivated to share with their friends to earn free products and coupons {\cite{li2022diffusion}.}
As promotion information becomes outdated over time, Pinduoduo periodically selects different users as seeds to update promotions to their friends and friends' friends in the social network timely. Thereby, we are motivated to optimize a practical multi-stage seeding process of viral marketing to keep promotion information that is received by users in social networks as fresh as possible. 

To evaluate information freshness from receivers' perspectives, age of information (AoI), which is coined in \cite{kaul2012real}, is
widely adopted as a standard performance metric in the literature  \cite{tripathi2022optimizing,yates2021age,talak2017minimizing,hsu2017age,xu2020peak}. AoI measures
the time elapsed since the latest information reached its
intended user.
In the literature, most AoI works either study broadcasting networks where a base station sends time-sensitive information
to its clients (e.g., \cite{bastopcu2019minimizing,hsu2017age,kadota2018scheduling,liu2021minimizing}), or focus on monitoring networks where clients coordinate to transmit their collected information to the base station timely (e.g., \cite{kadota2018optimizing,bedewy2019age,li2022scheduling}).  
In the context of broadcasting networks, Hsu et al. \cite{hsu2017age} proposed an MDP-based scheduling algorithm to minimize the long-run average AoI for noiseless channels, and Kadota et al. \cite{kadota2018scheduling}
minimized the expected weighted sum AoI of the clients for unreliable channels via transmission scheduling. In the context of monitoring networks, Tripathi et al. \cite{tripathi2021age} studied
a mobile agent's randomized trajectory to mine the data from ground terminals and optimize average AoI.  
In \cite{tripathi2022optimizing}, Tripathi et al. further studied 
the correlation among multiple coupled sources and provided an approximation solution for minimizing the weighted-sum average AoI. 
Besides average AoI, peak AoI also emerges as another important measure that quantifies the worst case AoI in a fair manner \cite{tripathi2021age,xu2020peak}. In priority queueing systems where a data source sends updates to a single processor, Xu and Gautam \cite{xu2020peak} derived closed-form expressions of the peak AoI, which further allow for analyzing the effects of specific service strategies. It is clear that prior AoI literature does not study fresh information diffusion in a general social network or optimize any seeding strategies over time. 

Additionally, recent works \cite{yates2021agegossip,buyukates2022version} consider the version age metric to assess information freshness in gossip networks, where each update at the source is treated as a version change and the version age indicates how many versions the information at the monitor is outdated. Unlike the original AoI metric, the version age of a monitor remains unchanged in between the version changes at the source. These studies typically consider a single source that is predefined and focus on characterizing the version age in various network structures (e.g., ring and fully connected networks \cite{yates2021agegossip} and community-like network \cite{buyukates2022version} where receiver nodes are grouped into equal-sized clusters). In contrast, our work aims to develop a multi-round seeding strategy to minimize either peak AoI (for full market coverage) or the average AoI (for spreading promotions to the majority of users), which better aligns with viral marketing's goal of increasing the likelihood of users' purchasing promoted products or services.

Our multi-stage seeding problem is also different from the traditional social network literature about information diffusion, which only limits to one-time seeding and overlooks AoI dynamics or information replacement over time. When information diffusion meets social networks, Bakshy et al. \cite{bakshy2012role} found through their extensive field experiments that users are likely to spread information via social networks. 
Lu et al. \cite{lu2015towards} proposed heuristic algorithms that achieve good performances in certain scenarios.  Ioannidis et al. \cite{ioannidis2009optimal} studied the problem of dynamic content dissemination in a complete contact (social) graph among users. Yet our social graph is not restricted to being complete and our problem needs to consider the AoI process of each node over the entire time horizon, 
rather than only at the final moment. These factors make our problem and analysis more intricate. For a more comprehensive understanding of social information diffusion, we refer interested readers to survey works \cite{bartal2021role,banerjee2020survey}.  The most related classical problem to ours is the NP-hard social influence maximization (SIM) problem \cite{kempe2003maximizing}, which seeks a subset $S\subseteq V$ of cardinality $k$ at a time
to maximize the influence spread function $\sigma(S)$ of the given social graph $G=(V, E)$. The submodularity of the objective function yields an approximation ratio of around $1-\frac{1}{e}$ 
\cite{nemhauser1978analysis,kempe2003maximizing}, where $e$ is the base of natural logarithm. However, SIM overlooks multi-round seeding and information propagation within the social network (where new information replaces outdated one) and fails to account for the information freshness experienced by users due to dynamic updates. These simplify both the algorithm design and approximation analysis there. Other related works include the \textit{$k$-median} \cite{cohen2022improved}  and the \textit{$k$-center} problems \cite{panigrahy1998ano,lu2015towards}.
For symmetric \cite{hochbaum1985best} and asymmetric \cite{panigrahy1998ano} graphs, the $k$-center problem achieves an approximation of a constant and an iterated logarithm of the number of nodes, respectively. Since different sequences of dynamically selected seeds make a huge difference in our AoI problem, all the above approximation algorithms and results turn out to be infeasible for our problem.

Our key novelty and main contributions are outlined below.
\begin{itemize}
    \item \textit{Optimizing the Age of Information Diffusion on Social Networks via Multi-stage Seeding Process.} To the best of our knowledge, our work is the first to model and optimize a multi-stage information seeding process on social networks. We practically allow fresh information to replace any outdated information during the network diffusion, and our multi-stage seeding process over any social network topology makes it intricate to trace the network AoI dynamics. By considering two distinct objectives, i.e., the peak and the average AoI of the network, we comprehensively study the dynamic optimization problem.  
    \item \textit{Closed-form Characterization of AoI Tracing and NP-hardness Proofs}. As a critical step, we successfully derive closed-form expressions that trace the average and the peak AoI objectives of any social network, respectively, which is highly non-trivial due to the dynamics in the multi-stage seeding process. By non-trivial reductions from the dominating set and set cover problems, respectively, we prove that both of our problems for peak and average AoI minimization are NP-hard.
      \item \textit{Fast Algorithm for peak AoI minimization with Provable Approximation Guarantee}. By focusing on a fine-tuned set of seed candidates along the social graph diameter, we design a new polynomial-time algorithm that guarantees good approximations as compared to the optimum. Particularly, our algorithm is proven to achieve an approximation of less than 2 for linear network topology. For a general network topology, we equivalently reduce it to a special histogram structure and analytically provide a provable approximation guarantee.
\item \textit{Fast Algorithm for average AoI minimization with Provable Approximation Guarantee.}  As the average AoI explicitly takes into account every user's AoI dynamics, it is more involved to optimize. We further provide two-sided bound analysis to build up a new framework for average AoI approximation analysis. This enables us to link our problem to a remarkably simplified sum-distance minimization problem and helps design a polynomial-time algorithm that achieves a good approximation guarantee. Finally, we validate our theoretical results via extensive experiments on a realistic social network.
\end{itemize}

The rest of this paper is organized as follows. Section~\ref{sec_problem_statement} gives the system model and problem statement. In Section~\ref{sec_aoi_formulation}, we derive closed-form expressions of our objectives and prove the NP-hardness of our problem. Section~\ref{sec_peak_aoi} and Section~\ref{sec_average_aoi} provide our algorithm design and approximation guarantee for our minimization problem under the peak and average AoI objectives, respectively. Finally,  our simulation results are summarized in Section~\ref{sec_experiments}.

%
\section{System Model and Problem Statement}\label{sec_problem_statement}
We consider a budget-aware viral marketing platform (e.g.,  Facebook Marketplace and Pinduoduo) that periodically selects one user as a seed in each of $k$ seeding rounds to dynamically diffuse the latest promotion updates over the social network. In other words, a new seed is selected in each round to propagate the updated promotion through the social network, allowing fresher promotions to replace older ones at each node.
Fig.~\ref{illustration_graph} displays the system configuration.
\begin{figure}[t]
    \centering
\includegraphics[width=6cm]{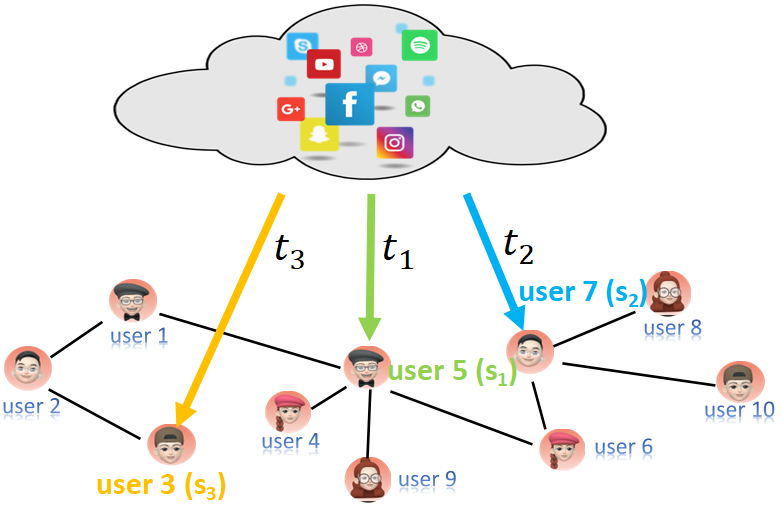}
    \caption{An illustrative example of the viral marketing system where the time gap between two consecutive seeding timestamps is $\Delta=2$. Here, we consider three rounds of seeding as $(s_1,s_2,s_3)=(v_5,v_7,v_3)$ to dynamically choose users 5, 7, 3 at time $t_1=1$, $t_2=3$ and $t_3=5$ for information diffusion.}
    \label{illustration_graph}
\end{figure}

Formally,
we model the social network as an undirected and connected graph $G=(V, E)$, where a node $v\in V$ represents a distinctive user in the network and an edge $(v_i,v_j)\in E$ tells a social connection between two users $v_i$ and $v_j$ to share messages. 
Suppose the network involves $n$ users and $m$ social connections, i.e., $|V|=n$ and $|E|=m$. We consider discretely slotted time horizon of $T$ slots and index current time by $t\in \{0,1,..., T\}$. 

Denote $S_k=(s_1,...,s_k)$ as the sequence of $k$ seeds that the platform dynamically selects from $V$ to diffuse the latest promotion information within the time horizon. Since many viral marketing promotions are practically scheduled 
in a periodic way
\cite{raja2012study}, we denote the time gap between any two consecutive promotion updates as $\Delta$. That is, at each timestamp $t_j=(j-1)\cdot\Delta+1$ with $j\in \{1,...,k\}$, the platform selects a user as a seed for viral marketing. 
Table~\ref{notationinouralgorithm} summarizes the key notation of this paper.

Now, we introduce the process of social information diffusion and our AoI models under multi-stage seeding. Without loss of generality, we normalize the minimum AoI value to one representing the smallest age of information, as in \cite{wang2022dynamic,tripathi2022optimizing}. 
For any user $v_i\in V$ at time $t$, its real-time AoI is denoted as $A(v_i,t)$ and evolves as follows.
\begin{enumerate}
    \item When $t=0$, the AoI of each user $v_i\in V$ is denoted to be a constant $A_{0}$, i.e., $A(v_i,0) = A_0$. Note that we can easily extend our solution to different initial ages, which will be shown via our experiments in  Section \ref{sec_experiments}.
    \item  At
    timestamp $t_j$ when the $j$-th seed $s_j$ is selected as the source to propagate the $j$-th latest promotion information, the AoI of $s_j$ drops to 1 immediately, i.e.,  $ A(s_j,t_j)=1$. This means that at time $t_j$, the $j$-th promotion at seed $s_j$ is the freshest among all promotions propagated up to that time. 
    \item  Each seed node will 
disseminate the new promotion to its neighbors in the social network \cite{bakshy2012role}. Upon receiving the update, each neighbor will then propagate the information to her own neighbors, and this process continues.  For ease of exposition, we suppose that information dissemination through each social connection (i.e., edge) consumes a normalized unit of time. 
\item  At any given time, if a node receives multiple information updates from its neighbors, it takes the freshest one and updates its AoI accordingly to align with the age of the taken information; otherwise, the AoI of the node will 
increase linearly over time.
\end{enumerate}
Let us consider an illustrative example of information diffusion and AoI evolution  in Figure \ref{illustration_graph}: 
user 5 is firstly selected as a seed at timestamp $t_1=1$, resulting in its AoI to be $A(v_5,t_1)=1$ (i.e., user 5's AoI at time $t_1$ is updated to 1). The updated information in green (which is from seed user 5) will reach user 5's social neighbors (i.e., users 1, 4, 6, 9) after a unit time slot, and the AoI of these neighbor users will be updated to two (which is the age of the green information at time $2$) at time $t_1+1=2$. User 5's social neighbors will further disseminate the information to their neighbors accordingly. Furthermore, when some new information reaches a node, the node will replace its old information with the new one. For instance, at time $t_1+2=3$, the green information will reach user 7 by disseminating from user~6. In the meantime of $t_2 = 1+\Delta =3$, user~7 is selected as a new seed to disseminate the new information in blue. Consequently, user~7 will replace its green information with the blue one, yielding $A(v_7,t_2)=1$. 

Due to the complexity of our multi-stage seeding process and the social network topology, directly tracing and analyzing the interrelated AoI dynamics among network nodes is challenging. To address this, we propose a new method that accurately tracks the AoI dynamics of each network node in closed form later in Section~\ref{sec_aoi_formulation}.
Before that, we proceed with problem formulations.
\begin{table}[t]
    \caption{Key notation in this paper.}
\begin{tabular}{ll}
    \hline
    Notation & Physical meaning\\
    \hline
    $n\triangleq |V|$ & The order of $G$ or overall user number.\\
  $m\triangleq |E|$ & The edge size of $G$ or overall social connections.\\
$dist(u,v)$ & The shortest distance between nodes $u$ and $v$ in $G$.\\
$diam(G)$& The diameter of the graph $G$ or diameter path.\\ 
    $T$  &Time horizon considered in this paper.\\ 
        $k$& Number of seeds to be selected within $T$ slots.\\
        $\Delta$ & Time gap between consecutive seeding time.\\
       $t_i$& The $i$th timestamp to select the $i$th seed.\\
    $s_i$& The $i$th seed which is selected at timestamp $t_i$.\\
    $S_i\triangleq (s_1,...,s_i)$& The first $i$ selected seeds.\\
 $A(v_i,t)$ & The real-time AoI of a node $v_i\in V$ at time $t$. \\
 \hline
    \end{tabular}
    \label{notationinouralgorithm}
\end{table}

We examine two optimization objectives comprehensively, focusing on fairness and efficiency: the 
peak and the average AoI, both of which are evaluated over the time horizon $[0,T]$. 

The peak AoI of the network $A^{\rm peak}$ is defined as 
\begin{equation} \label{def_peak_aoi_network}
    A^{\rm peak}= \max\limits_{v_i\in V}\{A_i^{\rm peak}\},
\end{equation}
in which $A_i^{\rm peak}\triangleq \max\limits_{ t\in[0,T]}\{A(v_i,t)\}$ denotes the peak AoI of node $v_i\in V$ over the time horizon. 

The average AoI of the network $A^{\rm avg}$ is defined as 
\begin{equation}\label{def_avg_aoi_network}
    A^{\rm ave} = \frac{1}{n}\sum\limits_{v_i\in V} A_i^{\rm ave},
\end{equation}
where $A_i^{\rm ave}\triangleq \frac{1}{T}\int_{t=0}^T A(v_i,t)dt$ denotes the average AoI of node $v_i\in V$ over the time horizon.  

Accordingly, we aim to minimize the following two objectives via multi-stage seeding with $(s_1,...,s_k)$, respectively:
 
\textbf{Objective 1}: \textit{Peak AoI Minimization Objective}: 
\begin{equation}\label{problem_peak_aoi}
   \min\limits_{(s_1,...,s_k)} \max\limits_{v_i\in V}\max\limits_{0\leq t\leq T}\{A(v_i,t)\}.
\end{equation}

\textbf{Objective 2}:
\textit{Average AoI Minimization Objective}:
\begin{equation}\label{problem_average_aoi}
       \min\limits_{(s_1,...,s_k)} \frac{\sum\limits_{v_i\in V} \int_{t=0}^T A(v_i,t)dt}{n\cdot T}.
\end{equation}

\begin{remark} 
To boost revenue, a viral marketing platform typically seeks to expose users to more promotions earlier, thereby increasing the likelihood of purchasing promoted products or services \cite{ailawadi2006promotion}. In this context, our peak and average AoI objectives (\ref{problem_peak_aoi})-(\ref{problem_average_aoi}) align well with the common goal of viral marketing platforms to maximize revenue for two reasons: First, a smaller AoI for a particular user at any given time directly indicates that the user has received the updated promotion, which the platform prefers to be spread to the public over outdated ones; Second,  a smaller AoI for a particular user at any given time also indicates that the user has received more promotions through the social network.  
\end{remark}

In this paper, we aim at tractable algorithms that guarantee good peak/average AoI performance in the worst-case scenario, which is evaluated via a standard metric \textit{approximation ratio}. Given an instance $I(G,T, \Delta, A_{0})$ of the problem, let $A^{\rm peak}_{\rm ALG}(I)$ and $A^{\rm peak}_{\rm OPT}(I)$ represent the peak AoI performance generated by our algorithm (ALG) and an optimal solution (OPT), respectively, for the same instance $I$. The approximation ratio of our algorithm ALG for the peak AoI  minimization problem in (\ref{problem_peak_aoi}), denoted as $\gamma_{\rm ALG}^{\rm peak}$, is defined as the supremum of the ratio of ALG's peak AoI over OPT's peak AoI among all possible instances $I$, i.e., 
\begin{equation}\label{peak_aoi_def}
    \gamma_{\rm ALG}^{\rm peak}=\sup\limits_{I(G,\Delta)}\frac{A^{\rm peak}_{\rm ALG}(I)}{A^{\rm peak}_{\rm OPT}(I)}.
\end{equation}
Denote $A^{\rm avg}_{\rm ALG}(I)$ and $A^{\rm avg}_{\rm OPT}(I)$ as the average AoIs generated by ALG and OPT, respectively, for the same instance $I$. Similar to (\ref{peak_aoi_def}), the approximation ratio for the average AoI minimization problem in (\ref{problem_average_aoi}) is defined as
\begin{equation}\label{avg_aoi_def}
    \gamma_{\rm ALG}^{\rm avg}=\sup\limits_{I(G,\Delta)}\frac{A^{\rm avg}_{\rm ALG}(I)}{A^{\rm avg}_{\rm OPT}(I)}.
\end{equation}
Then, our aim is to pursue fast (polynomial-time) algorithms that guarantee good approximation ratios of small (\ref{peak_aoi_def}) and (\ref{avg_aoi_def}) for the two problems (\ref{problem_peak_aoi}) and (\ref{problem_average_aoi}), respectively.

Since the AoI optimization objectives in problems (\ref{problem_peak_aoi}) and (\ref{problem_average_aoi}) have not yet been quantified, our first critical step is to quantify AoI objectives in the next section. Following that, we will provide a theoretical foundation by proving the NP-hardness of our problems and designing approximation algorithms. 
\section{AoI Tracing and NP-hardness}\label{sec_aoi_formulation} 
To derive the exact AoI objective expressions necessary for analyzing problems (\ref{problem_peak_aoi}) and (\ref{problem_average_aoi}), we must consider the real-time AoI dynamics of each node. Based on our information diffusion model described in Section~\ref{sec_problem_statement}, the AoI $A(v_i,t)$ of a node $v_i\in V$ drops only when $v_i$ is updated by some of its neighbors with fresher information at time $t$, and increases linearly otherwise.  Note that, at any given time, the real-time AoI of all nodes holding the same information remains the same. This implies that the age of any updated information can be traced back to the seed of that information. 

Hence, it is crucial to determine which seed in $S_k=(s_1,...,s_k)$ contributes to the decrease of node $v_i$' AoI at time~$t$. Specifically, any such seed, say $s_x$, satisfies the following two conditions (\ref{condition_seed_01})-(\ref{condition_seed_02}).
First, the information update from seed $s_x$ reaches $v_i$ exactly at time $t=t_x+dist(s_x,v_i)$,  where $dist(s_x, v_i)$ represents the shortest time required to propagate the $x$-th promotion information to node $v_i$. According to our model in Section~\ref{sec_problem_statement}, this duration $dist(s_x, v_i)$ is numerically equal to the shortest distance between nodes $s_x$ and $v_i$ in the social graph~$G$. This leads to the following condition (\ref{condition_seed_01}), which holds since $dist(s_x,v_i)\geq 0$, implying that $s_x$ is selected no later than time~$t$.
\begin{equation}\label{condition_seed_01}
        1\leq t_x\leq t.
    \end{equation}
Second, the new information propagated from seed $s_x$ must be fresher than any prior information that $v_i$ has received before time $(t-1)$, which results in the following condition (\ref{condition_seed_02}).
\begin{equation}\label{condition_seed_02}
    1+t-t_x\leq A(v_i,t-1).
\end{equation}
Since multiple seeds may simultaneously meet the conditions (\ref{condition_seed_01})-(\ref{condition_seed_02}) at time $t$, we define the following set $\Omega(v_i,t)$ 
to summarize all such seeds that could lead to $v_i$'s information update at time $t$:
\begin{equation*}
\begin{split}
 &\Omega(v_i,t)\\
 &\triangleq{\{s_x|1\leq t_x\leq t,\;t-t_x=dist(v_i,s_x)\leq  A(v_i,t-1)-1\}}.
 \end{split}
\end{equation*}

In light of $\Omega(v_i,t)$, we can trace and formulate $v_i$'s AoI over time $[0,T]$ directly by the following Lemma \ref{lemma_aoiof_vi},
where user $v_i$ will update its AoI at time $t$ by the freshest information it receives up to time $t$.
\begin{lemma}\label{lemma_aoiof_vi}
Given $S_k={s_1,...,s_k}$ as the set of selected seeds, the following holds, where $\Omega$ specifies $\Omega(v_i,\left \lfloor t \right \rfloor)$:
\begin{equation*}
\begin{split}
&A(v_i,t)=\\
    &\left\{\begin{matrix}
    A(v_i,0)+t,&{\rm\;if\;}t\in[0,1),\\
    \min\limits_{s_x\in\Omega}\{1+dist(v_i,s_x)\}+t-\left \lfloor t \right \rfloor, &{\rm\;if\;} \Omega\neq \varnothing,t\in [1,T],\\
A(v_i,\left \lfloor t \right \rfloor-1)+1+t-\left \lfloor t \right \rfloor, &{\rm\;if\;}\Omega=\varnothing, t\in [1,T].\\ 
\end{matrix}\right.
\end{split}
\end{equation*}

\end{lemma} 
Although Lemma~\ref{lemma_aoiof_vi} provides some insights into the AoI evolution over network and time domains, it cannot be directly applied to estimate either peak AoI (\ref{problem_peak_aoi}) or the average AoI (\ref{problem_average_aoi}).  To identify the specific time points at which $v_i$'s AoI drops, we introduce the following definition of \textit{discontinuity points} for each node $v_i$.
\begin{definition}[Discontinuity point]\label{def_discontinuitypoint}
A time point $t$ is called a discontinuity point to node $v_i$ if $\Omega(v_i,t)\neq \varnothing$. In other words, at some time $t$, when $\Omega(v_i,t)\neq \varnothing$, node $v_i$ will experience a decrease in its AoI as a result of receiving a promotion updates from some of its neighbors in $\Omega(v_i,t)$.
\end{definition}

To identify those discontinuity points of a node $v_i$, we have the following Lemmas~\ref{prop_dicontinuity_01} and~\ref{prop_dicontinuity_02} for guidance. For convenience, we use $[k]$ to refer to the index set $\{1,...,k\}$. 
\begin{lemma}\label{prop_dicontinuity_01}
Given the sequence $S_k=(s_1,...,s_k)$ of dynamically selected seeds, each of  $v_i$'s discontinuity points can be found in set $\{t_x+dist(s_x,v_i)|x\in [k]\}$.
\end{lemma}
Intuitively, Lemma \ref{prop_dicontinuity_01} offers a ground set $\{t_x+dist(s_x,v_i)|x\in [k]\}$ from which we can search to find all possible discontinuity points to each user node $v_i\in V$. To precisely pin down those discontinuity points of a node $v_i$, 
let us consider
two seeds, say $s_j$ and $s_y$, where seed $s_y$'s information is fresher than seed $s_j$'s and reaches node $v_i$ 
earlier than $s_j$'s diffusion. It is clear that $v_i$ will update its AoI by $s_y$'s information rather than $s_j$'s. Based on this observation, we have the following lemma to help us identify non-discontinuity points within the ground set $\{t_x+dist(s_x,v_i)|x\in [k]\}$.
\begin{lemma}\label{prop_dicontinuity_02}
An element $t_j+dist(s_j,v_i)\in\{t_x+dist(s_x,v_i)|x\in [k]\}$ is not $v_i$'s discontinuity point, if 
there exists some other element  $t_y+dist(s_y,v_y)\in \{t_x+dist(s_x,v_i)|x\in [k]\}$) which satisfies $t_y>t_j$ and $t_y+dist(s_y,v_y)\leq t_j+dist(s_j,v_i)$.
\end{lemma}

\begin{algorithm}
\caption{Discontinuity Point Finder}\label{DiscontinuityPoint}
 \textbf{Data}: {$v_i\in V$, $S_k=(s_1,...,s_k)$.}\\
  \textbf{Result}: {discontinuity points of $v_i$, i.e., $(s_{i_1},...,s_{i_{k_i}})$.}\\
  \begin{algorithmic}[1] 
 \STATE Initialize $U_i\leftarrow\{t_x+dist(s_x,v_i)|x\in [k]\}$.\\ 
    \FOR{$p= 1,...,k-1$}
    \FOR{$q=p+1,...,k$}
\IF{$t_q+dist(s_q,v_i)\leq t_p+dist(s_p,v_i)$}
\STATE $U_i\leftarrow U_i-\{t_p+dist(s_p,v_i)\}$.\\
    \textbf{Break}.
    \ENDIF
    \ENDFOR
    \ENDFOR
\end{algorithmic}
\end{algorithm}

Based on our Lemmas~\ref{prop_dicontinuity_01} and~\ref{prop_dicontinuity_02}, a straightforward approach to finding all discontinuity points of each node $v_i\in V$ is to remove from the ground set ${t_x+dist(s_x,v_i)|x\in [k]}$ all those non-discontinuity points of $v_i$ one by one. Accordingly, we have Algorithm~\ref{DiscontinuityPoint} which runs in $O(k^2)$ time due to its dominating Steps 2-9 in its nested for-loop. 

However, we find this is not sufficiently efficient.  In the following, we propose a new Algorithm~\ref{ImprovedDiscontinuityPoint} that finds discontinuity points, which are due to the conditions (\ref{condition_seed_01})-(\ref{condition_seed_02}), backward-sequentially in the $\{t_x+dist(s_x,v_i)|x\in [k]\}$ of Lemma~\ref{prop_dicontinuity_01}. This process can be achieved by the non-nested while-loop in Algorithm~\ref{ImprovedDiscontinuityPoint}, which further accelerates the finding by reducing the time complexity significantly from $O(k^2)$ to $O(k)$. 

%
\begin{algorithm}
\caption{Accelerated Discontinuity Finder}\label{ImprovedDiscontinuityPoint}
\textbf{Data:} {$v_i\in V$, $S_k=\{s_1,...,s_k\}$.}\\
  \textbf{Result}: {Stack $U_i$ of $v_i$'s discontinuity points.}\\
\begin{algorithmic}[1]
\STATE \textbf{Initialization.}   $W_i\leftarrow\{t_x+dist(s_x,v_i)|x\in [k]\}$.
\STATE ${\rm Pop}(W_i,\tau_{\rm up})$ and $U_i\leftarrow\{\tau_{\rm up}\}$.
   \WHILE{$W_i\neq \varnothing$}
     \STATE ${\rm Pop}(W_i,\tau_{\rm down})$.
 \IF{$\tau_{\rm down}< \tau_{\rm up}$}
\STATE ${\rm Push}(U_i, \tau_{\rm down})$.
\STATE $\tau_{\rm up} \leftarrow \tau_{\rm down}.$
\ENDIF
\ENDWHILE
\end{algorithmic}
\end{algorithm}

Concretely, our Algorithm~\ref{ImprovedDiscontinuityPoint} first initializes a stack $W_i$ by pushing elements $\{t_x+dist(s_x,v_i)|x\in [k]\}$ sequentially in increasing order of their sub-index $x$, ensuring that the top element in $W_i$ always holds the freshest information among all the others in $W_i$. Algorithm~\ref{ImprovedDiscontinuityPoint} maintains two stacks $W_i$ and $U_i$, which contain candidate discontinuities and de facto discontinuities regarding user $v_i$, respectively. Within $W_i$ (resp. $U_i$), an upper (resp. a lower) element corresponds to a fresher information source. Moreover, the lower element within $U_i$ corresponds to a larger value, i.e., a later timestamp. Since the $k$th information source from seed $s_k$ is the freshest ever, its corresponding point $t_k+dist(s_k,v_i)$ is a de facto discontinuity point for $v_i$, and thus will be firstly included in $U_i$, i.e.,  $U_i\leftarrow \{t_k+dist(s_k,v_i\}$ in Line 2 of our Algorithm~\ref{ImprovedDiscontinuityPoint}.

In addition, Algorithm~\ref{ImprovedDiscontinuityPoint} also maintains two variables $\tau_{\rm up}$ and $\tau_{\rm down}$ that store the latest discontinuity point that has been found and included in $U_i$ and the top element in stack $W_i$, respectively. In each iteration of its while loop, Algorithm~\ref{ImprovedDiscontinuityPoint} pops the top element from the latest stack  $W_i$ and assigns its value to $\tau_{\rm down}$, i.e., $\tau_{\rm down}= \arg\max\limits_{t_x+dist(s_x,v_i)\in W_i}\{x\}$. Note that this $\tau_{\rm down}$ corresponds to the freshest information source within the current $W_i$, which will lead to an AoI drop for $v_i$ if it reaches $v_i$ earlier
than all those previously found discontinuities in $U_i$, which is indicated by the condition $\tau_{\rm down}<\tau_{\rm up}$. 
Algorithm~\ref{ImprovedDiscontinuityPoint} iterates until finding all discontinuity points of $v_i$, i.e., when $W_i=\varnothing$, which can complete in $O(k)$-time due to its sole while loop. We illustrate in Fig.~\ref{fig_improveddiscontinuityfinder} an execution example of Algorithm~\ref{ImprovedDiscontinuityPoint}  for ease of understanding.
\begin{figure}[t]
    \centering  \includegraphics[width=9cm]{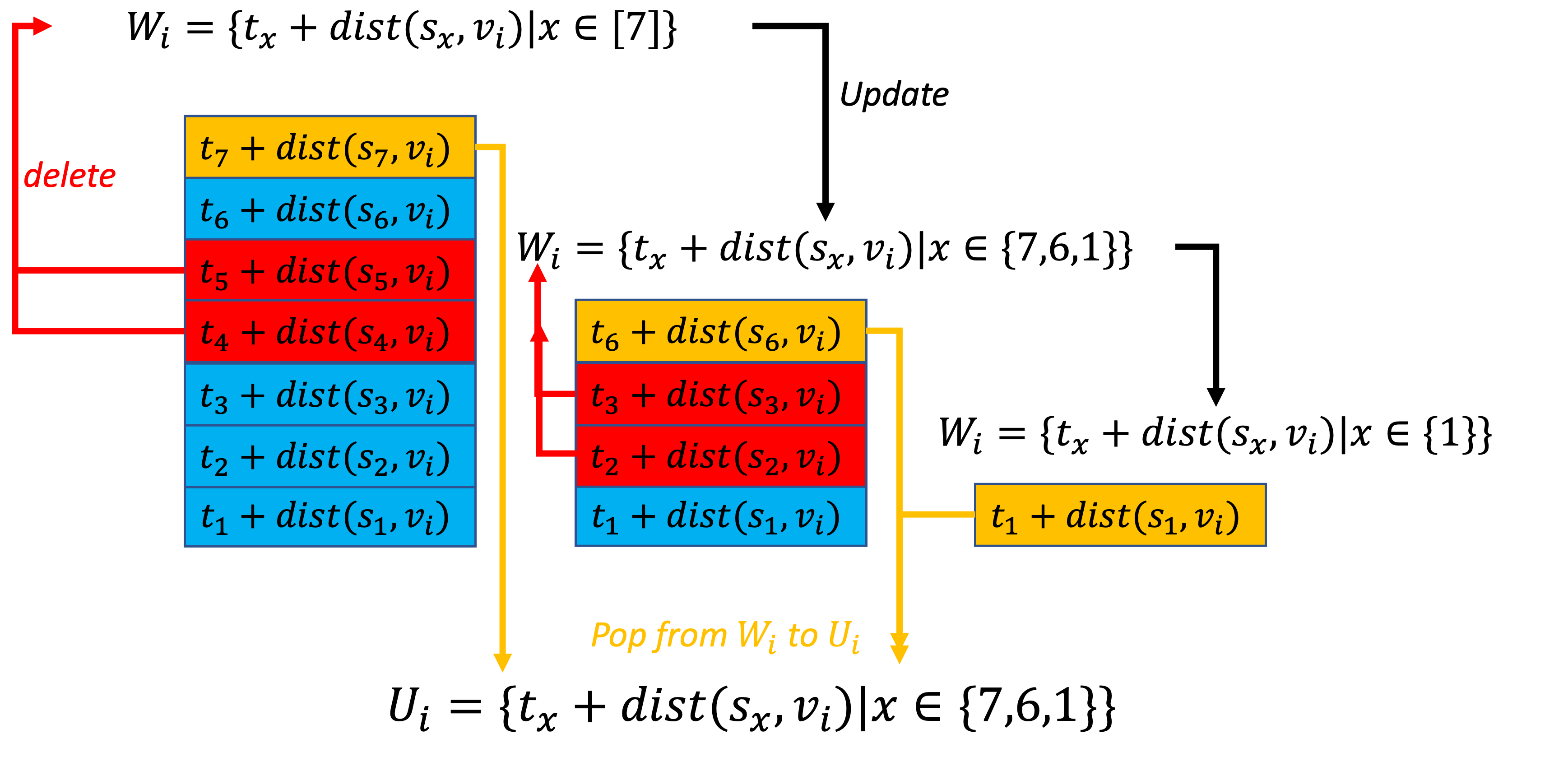}
    \caption{Illustration example of Algorithm~\ref{ImprovedDiscontinuityPoint}, where $k=7$, $t_1+dist(s_1,v_i)<t_6+dist(s_6,v_i)<t_2+dist(s_2,v_i)<t_3+dist(s_3,v_i)<t_7+dist(s_7,v_i)<t_4+dist(s_4,v_i)<t_5+dist(s_5,v_i)$.}
\label{fig_improveddiscontinuityfinder}
\end{figure}


%
With discontinuity points returned by Algorithms~\ref{DiscontinuityPoint} and~\ref{fig_improveddiscontinuityfinder}, we are now ready to characterize the expressions of the peak and average AoI objectives in the following subsection. 
\subsection{Closed-Form Characterization of Peak and Average AoI Objectives}\label{subsection_closedform_aoi}
For node $v_i\in V$, let $k_i$ represent the overall number of discontinuity points in $A(v_i,t)$ over the time horizon $[0,T]$. Denote $t_{i_j}$ as the selection timestamp of the seed that produces the $j$th discontinuity point of $v_i$ (in $A(v_i,t)$). The set of all discontinuity points of node  $v_i$ is given by
 $\{t_x+dist(s_x,v_i)|x\in \{i_1,...,i_{k_i}\}\}$. We note $i_{k_i}=k$ since the last seed $s_k$ carries the freshest promotion, resulting in a final AoI decrease for all other nodes in the network.
 
By taking into account the discontinuity points of $v_i$ in  $\{t_x+dist(s_x,v_i)|x\in \{i_1,...,i_{k_i}\}\}$ in Lemma~\ref{lemma_aoiof_vi}, we have the following proposition. 
\begin{table*}[t]
\begin{equation}\label{eq_functionof_vi}
\begin{split}
A(v_i,t)=\left\{\begin{matrix}
A_{0}+t, &{\rm if\;}t\in[0,t_{i_1}+dist(s_{i_1},v_i)), \\ 
 1+t-t_{i_1},&{\rm if\;}t\in[t_{i_1}+dist(s_{i_1},v_i),t_{i_{2}}+dist(s_{i_2},v_i)),\\
 \vdots &\vdots\\
1+t-t_{i_{k_i-1}},&{\rm if\;}t\in[t_{i_{k_i-1}}+dist(s_{i_{k_i-1}},v_i),t_{i_{{k_i}}}+dist(s_{i_{k_i}},v_i)), \\
 1+t-t_{i_{k_i}},& {\rm if\;}t\in[t_{i_{k_i}}+dist(s_{i_{k_i}},v_i),T].
\end{matrix}\right.
\end{split}
\end{equation}
\end{table*}
\begin{proposition}\label{theorem_functionof_vi}
Given the sequence $S_k=(s_1,...,s_k)$ of dynamically selected seeds, $A(v_i,t)$ is piece-wise over time $t\in [0,T]$ and is given by (\ref{eq_functionof_vi}) where $(s_{i_1},...,s_{i_{k_i}})$ are obtained by running Algorithm~\ref{DiscontinuityPoint} on $S_k$,
\end{proposition}
%
According to Proposition \ref{theorem_functionof_vi}, the new AoI dynamics in (\ref{eq_functionof_vi}) depends solely on discontinuity points returned by our Algorithm~\ref{DiscontinuityPoint}. This enables us to precisely trace the AoI pattern of each node. Particularly, in each time interval of (\ref{eq_functionof_vi}), the expression of $A(v_i,t)$ becomes regular, which paves the way for us to characterize the peak and average AoI objectives in closed-form below. 
Consequently, we have the following Theorems~\ref{lemma_aoi_peak_v_i} and~\ref{avgaoi_formulation_lemma}, respectively.
\begin{theorem}\label{lemma_aoi_peak_v_i}
Given the sequence $S_k$ of dynamically selected seeds, the peak AoI objective of the network can be expressed in closed form as follows:
\begin{equation}\label{eq_peak_vi}
\begin{split}
    A^{\rm peak}=\max\limits_{v_i\in V}\{&A_{0}+t_{i_1}+dist(s_{i_1},v_i),1+T-t_{i_{k_i}},\\
    &
    \max\limits_{j\in\{2,...,{k_i}\}}\{1+dist(s_{i_{j}},v_i)+t_{i_j}-t_{i_{j-1}}\}\}.
\end{split}
\end{equation}
\end{theorem}
\begin{theorem}\label{avgaoi_formulation_lemma}
Given the sequence $S_k$ of dynamically selected seeds, the average AoI objective of the network can be expressed in closed form as follows:
\begin{equation}\label{avgaoi_formulation}
    A^{\rm avg}=\frac{1}{nT}\sum\limits_{i=1}^n\sum\limits_{j=1}^{k_i+1} \frac{(2A_{i{j-1}}+\Lambda_{ij})\cdot \Lambda_{ij}}{2},
\end{equation}
where
\begin{equation}\label{formulation_Aij}
    A_{ij}=\left\{\begin{matrix}
A_{0}, & j=0, \\ 
1+dist(s_{i_{j}},v_i) & j=1,...,k_i,\\ 
\end{matrix}\right.
\end{equation}
and
\begin{small}
\begin{equation*}
\begin{split}
&\Lambda_{ij}\\
&=\left\{\begin{matrix}
1+(i_1-1)
\cdot\Delta+dist(s_{i_1},v_i),& j=1,\\ 
 T-1-(i_{k_i}-1)
\cdot\Delta-dist(s_{i_{k_i}},v_i),& j=k_i+1,\\
(i_j-i_{j-1})\Delta+dist(s_{i_j},v_i)-dist(s_{i_{j-1}},v_i),&{\rm others}.
\end{matrix}\right.
\end{split}
\end{equation*}
\end{small}
\end{theorem}
\subsection{NP-hardness for Peak and Average AoI Minimization Problems}
The above Theorems~\ref{lemma_aoi_peak_v_i} and~\ref{avgaoi_formulation_lemma} are essential to analyzing our problems  (\ref{problem_peak_aoi}) and (\ref{problem_average_aoi}) later in Sections \ref{sec_peak_aoi} and \ref{sec_average_aoi}, respectively. Prior to that, we now construct a theoretical foundation for our problems by showing that they are NP-hard indeed.
%
%
%
%
\begin{theorem}\label{theorem_peak_aoi}
The peak AoI minimization problem  (\ref{problem_peak_aoi}) is NP-hard.
\end{theorem}
\begin{proof}
We show the NP-hardness of the peak AoI minimization problem by a highly non-trivial reduction from the NP-hard \textit{dominating set} problem in decision version \cite{garey1979computers}. Given an undirected graph $G(V,E)$ and a constant integer $k$, the decision version of the \textit{dominating set} problem answers ``yes" if and only if there exists a subset $S\subseteq V$ of size $k$ such that each node $v\in V$ is either in $S$ or is a neighbor of some $s\in S$. For example, if we consider the graph in Fig.~\ref{illustration_graph} with $k=3$, user nodes 2, 5, and 7 form a de facto dominating set. 

Given an instance of the dominating set problem, we reduce it to our peak AoI minimization problem (where we set a large $A_0=T+diam(G)$ and a small $\Delta\rightarrow 0$) on the same graph example $G=(V,E)$ with budget $k$.  The weight of each edge in $E$ is set as 1, i.e.,
\begin{equation}
   dist(v_x,v_y)=1, {\rm\;if\;} (v_x,v_y)\in E.
\end{equation}
As we set $\Delta\rightarrow 0$, all the $k$ seeds could be regarded as being selected together at time $1$. As a consequence, 
\begin{equation}\label{eq02_peak_nphardness}
s_{i_1}=\arg\min\limits_{s_j\in S_k}\{t_{j}+dist(s_j,v_i)\} =\arg\min\limits_{s_j\in S_k}\{dist(s_j,v_i)\},
\end{equation}
As we set $A_0=T+diam(G)$, Theorem~\ref{lemma_aoi_peak_v_i} further infers that the peak AoI of the graph now  depends on the first term of (\ref{eq_peak_vi}), i.e., 
\begin{equation}\label{eq01_peak_nphardness}
\begin{split}
A^{\rm peak}(S_k)=\max\limits_{v_i\in V}\{A_0+t_{i_1}+dist(s_{i_1},v_i)\}.
\end{split}
\end{equation}
By substituting (\ref{eq02_peak_nphardness}) into (\ref{eq01_peak_nphardness}), we obtain
\begin{equation}\label{formulation_peak_AoI_nphard}
\begin{split}
A^{\rm peak}(S_k)=A_0+1+\max\limits_{v_i\in V}\min\limits_{s_j\in S_k}\{dist(v_i,s_j)\}.
\end{split}
\end{equation}
The decision version of our peak AoI problem aims to answer whether the optimal solution is equal to or less than $A_0+2$. By applying a well-known binary search approach to solutions for our decision problem, one can obtain in polynomial-time a solution to our optimization problem. Reversely, given a solution to our optimization problem, one can directly obtain a solution to the decision problem. The above two facts reveal the 
equivalence relation between the decision and optimization versions of our AoI problem. Further, our NP-hardness can be proved readily with the following Lemma~\ref{prop04}. 
\begin{lemma}\label{prop04}
The dominating set problem answers ``yes" iff the decision version of our problem on the same input answers ``yes".
\end{lemma}
\begin{proof}[Proof of Lemma \ref{prop04}]
We discuss two cases.

\noindent\textbf{Case 1. }
(``$\Rightarrow$") 
When the dominating set problem answers ``yes", there exists a subset $S'_k\subseteq V$ with size $k$ such that 
\begin{equation}\label{eq_for_nphardness_peak_01}
    \max\limits_{v_i\in V}\min\limits_{s\in S'_k}dist(s,v_i)\leq 1.
\end{equation}
By applying (\ref{eq_for_nphardness_peak_01}) to (\ref{formulation_peak_AoI_nphard}), we have
$A^{\rm peak}\leq A_0+1+1=A_0+2$, telling that our AoI problem answers ``yes" as well.

\noindent \textbf{Case 2.}
(``$\Leftarrow$") 
When our AoI problem answers ``yes",  from (\ref{formulation_peak_AoI_nphard}) we have that
$\max\limits_{v_i\in V}\min\limits_{s\in S'_k}dist(s,v_i)\leq 1.$ This implies that each node is either in set $S_k$ or a neighbor of a node in $S_k$. That is, $\min\limits_{s\in S_k}dist(s,v)\leq 1$ holds for each node $v\in V$, telling that the dominating set problem answers ``yes".
\end{proof}
Lemma~\ref{prop04} concludes our proof of Theorem \ref{theorem_peak_aoi}.
\end{proof}

For the average AoI minimization, our closed-form expression in (\ref{avgaoi_formulation}) reveals the fact that each node's AoI dynamics truly affect the AoI objective. This makes our average AoI minimization problem more involved in comparison to the peak AoI minimization problem. By a subtle graph construction technique, our average AoI problem can be reduced from the NP-hard \textit{set cover problem}, yielding the following Theorem~\ref{theorem_nphardness_avg}.
\begin{theorem}\label{theorem_nphardness_avg}
The average AoI minimization problem (\ref{problem_average_aoi}) is NP-hard.
\end{theorem}


 
%
 
 \section{Peak AoI Minimization Algorithms}\label{sec_peak_aoi}

To solve the NP-hard peak AoI minimization problem (\ref{problem_peak_aoi}), we propose Algorithm~\ref{peak_cycliselection_alg} that seeds in a round-robin way
from a 
fine-tuned set of candidates along the diameter path,  which indicates the shortest path between the two most distant nodes in the social graph $G$ and can be found in $O((|V'|+m)^n)$-time by the Floyd-Warshall Algorithm \cite{cormen2022introduction}. In the sequel, we simply use diameter to refer to the diameter path when the context is clear.
 The rationale behind Algorithm~\ref{peak_cycliselection_alg}'s focus on seeding users on the diameter is twofold: \textit{first}, seeding along the diameter will directly expedite the information dissemination among nodes on the graph diameter, which usually takes a considerable time; \textit{second}, seeded users on the diameter are able to relay information to users in other branches, thereby improving the overall efficiency of information propagation within the entire network.
\begin{algorithm}[t]
\caption{\textsc{Cyclic Seeding} for peak AoI minimization}\label{peak_cycliselection_alg}
  \textbf{Data:} {$G(V,E)$, $k$.}\\
  \textbf{Result: }{$S_k=(s_1,...,s_k)$.}
  \begin{algorithmic}[1]
\STATE Index nodes on $diam(G)$ sequentially from one side to the other in set $V_{\rm diam} \triangleq(v_1,...,v_{|diam(G)|+1})$.
\STATE From $V_{\rm diam}$, select a sequence of seed candidates as $\Omega$.\\
\STATE Seed sequentially and recursively from $\Omega$.
  \end{algorithmic}
\end{algorithm}

Generally, our Algorithm~\ref{peak_cycliselection_alg} consists of the following three steps, where $diam(G)$ gives the diameter path of graph $G$.

\textbf{Step 1.} Index nodes on the diameter path $diam(G)$ sequentially from one side to the other, resulting in 
\begin{equation}
    V_{\rm diam}\triangleq(v_1,...,v_{|diam(G)|+1}),
\end{equation}
where $|diam(G)|+1$ gives the number of nodes on $diam(G)$. 

\textbf{Step 2.} Select from $V_{\rm diam}$ an integer number $\underline{\mu}$ of seed candidates, denoted as the sequence $\Omega \triangleq(\omega_1,...,\omega_{\underline{\mu}})\subseteq V_{\rm diam}$, in which the value of $\underline{\mu}$ and the specific locations of the candidates in $\Omega$ will be optimized later in Proposition \ref{prop_def_three_values}. 

\textbf{Step 3.}
Seed sequentially from $\Omega = (\omega_1,...,\omega_{\underline{\mu}})$ in the following round-robin manner: denote ${\rm mod}(x,y)$ as the number of $x$ modulo $y$, then, each seed $s_i\in S_k$ with $i\in \{1,...,k\}$ is selected to be $\omega_{{\rm mod}(i-1,\underline{\mu})+1}$.  
\begin{figure}[t]
    \centering    \includegraphics[width=8cm]{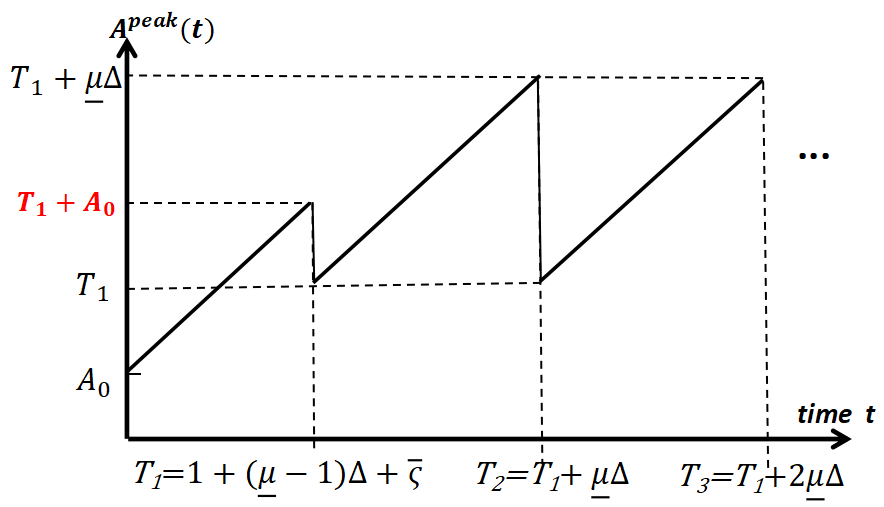}
    \caption{Peak AoI of Algorithm~\ref{peak_cycliselection_alg} in a line-type network. }\label{pettern_peak_aoi_cyclicselction}
\end{figure}

As Algorithm~\ref{peak_cycliselection_alg} seeds in a round-robin way from a well-designed set of seed candidates $\Omega$, one can expect a periodic pattern of the peak AoI dynamics {$A^{\rm peak}(t)=\max\limits_{v_i\in V}A(v_i,t)$}
after an initial time under Algorithm~\ref{peak_cycliselection_alg}. In fact, our Theorem~\ref{avg_approximation_theorem} later demonstrates that $A^{\rm peak}(t)$ follows a periodic pattern starting from time $T_1= 1+(\underline{\mu}'-1)\Delta+\overline{\varsigma}'$, as shown in Fig.~\ref{pettern_peak_aoi_cyclicselction}. Further, our following Proposition \ref{prop_def_three_values} shows that minimizing the peak AoI in (\ref{eq_peak_vi}) is equivalent to minimizing $T_1$, which is by choosing the number $\underline{\mu}'$ of selected seeds in Step 2 of Algorithm~\ref{peak_cycliselection_alg}. 
\begin{proposition}\label{prop_def_three_values} 
    In Algorithm~\ref{peak_cycliselection_alg}, the optimal ($\underline{\mu},\overline{\varsigma}$) should be chosen by solving the following problem:
    \begin{alignat}{2}
\min\limits_{(\overline{\varsigma}',\underline{\mu}')}\quad & T_1=1+(\underline{\mu}'-1)\Delta+\overline{\varsigma}'&\label{obj_peak_mu}\\
\mbox{s.t.}\quad
&\underline{\mu}'+\Delta(\underline{\mu}'-1)\underline{\mu}'+2\overline{\varsigma}'\underline{\mu}' \geq |diam(G)|+1, \quad& \label{enough_updates_constraint}\\
&  0\leq \overline{\varsigma}'<\Delta,&\\
 &\overline{\varsigma}',\underline{\mu}' {\rm \;are\;both\;integers},&\label{integerconstraint}
\end{alignat}
in which (\ref{enough_updates_constraint}) ensures that all nodes on $diam(G)$ are updated by time $T_1$. The optimal solution $(\overline{\varsigma},\underline{\mu})$ to Problem (\ref{obj_peak_mu})-(\ref{integerconstraint}) is
given by:
\begin{equation}\label{eq01_prop2}
\underline{\mu}=(\left \lfloor \frac{\Delta-1+\sqrt{\Delta^2+2\Delta+4|diam(G)|\Delta+1}}{2\Delta} \right \rfloor,
\end{equation}
and
\begin{equation}\label{eq02_prop2}
\overline{\varsigma}=\left \lceil \frac{|diam(G)|+1-(\underline{\mu}^2\Delta -\underline{\mu}\Delta+\underline{\mu})}{2\underline{\mu}} \right \rceil).
\end{equation}
\end{proposition}

In light of Proposition \ref{prop_def_three_values}, the optimal solution ($\underline{\mu},\overline{\varsigma}$) to 
Problem (\ref{obj_peak_mu})-(\ref{integerconstraint}) indicates that Algorithm~\ref{peak_cycliselection_alg} achieves the earliest time when each node on $diam(G)$ is updated. With (\ref{eq01_prop2}) and (\ref{eq02_prop2}) in hand, we can run Algorithm~\ref{peak_cycliselection_alg} optimally, finalizing the configuration of each seed candidate $\omega_{x}=v_{\zeta(x)}$ in its Step~2 as $\zeta(x)=\Delta(\underline{\mu}-x+1)^2+\underline{\mu}-x+1+(2\underline{\mu}-2x+1)\overline{\varsigma}$.
\subsection{Approximation Analysis in Line-type Networks}
To evaluate the approximation guarantee achieved by Algorithm~\ref{peak_cycliselection_alg}, we begin with a line-type (or linear) network topology as defined below. This simple yet fundamental topology serves as a basis for our approximation analysis, which will be extended to general social networks in Section \ref{subsec_general_graph}.
\begin{definition}[Line-type social network \cite{krawczyk2011line,yates2020age}]
A line-type social network consists of a set of sequentially connected users/nodes, 
which resembles a line and features an adjacency matrix as in (\ref{def_linear_network}).
\begin{equation}\label{def_linear_network}
E_{n\times n}=  
\begin{bmatrix}
0  &1  & ... &0 &0\\ 
1  &0  &...  &0 &0\\ 
\vdots &\vdots  & \vdots & \vdots&\vdots\\ 
 0 &0  &...  &0&1\\
 0&0  &... &1&0
\end{bmatrix}.
\end{equation}
\end{definition}
Given the NP-hard nature of our peak AoI minimization problem (\ref{problem_peak_aoi}), one can hardly find an optimal solution. Instead, 
we establish a lower bound approximation for the optimum solution by the following lemma, which draws upon our closed-form expression (\ref{eq_peak_vi}).


\begin{lemma}\label{theorem_peak_opt}
In line-type social networks, an optimal solution to the peak AoI minimization problem follows 
\begin{equation}\label{theorem_peak_opt_eq}
    A^{\rm peak}_{\rm OPT}\geq \max\{A_0+\overline{\varsigma},\left \lfloor {\xi} \right \rfloor\Delta\}+1+\underline{\mu}\Delta.
\end{equation}
in which $\xi=\frac{-(1+3\Delta)+\sqrt{(1+3\Delta)^2+4\Delta(1+2\underline{\mu}\Delta)}}{2\Delta}$.
\end{lemma}
Lemma \ref{theorem_peak_opt} provides a comparison benchmark for the approximation of our Algorithm~\ref{peak_cycliselection_alg}, yielding the following theorem.
\begin{theorem}\label{avg_approximation_theorem}
   For a line-type social network of (\ref{def_linear_network}), Algorithm~\ref{peak_cycliselection_alg} runs in $O(mn)$-time and guarantees an approximation $\frac{1+2\underline{\mu}\Delta+\overline{\varsigma}}{1+\underline{\mu}\Delta+\left \lfloor {\xi} \right \rfloor\Delta}$ of the optimum, which is strictly smaller than 2. Particularly, Algorithm~\ref{peak_cycliselection_alg} is optimal for $n\leq \frac{A_0^2+A_0(1-\Delta)}{\Delta}$. 
\end{theorem}
\begin{proof}
In a line-type social network, a middle node can disseminate information on both sides simultaneously. Next, we will show that the peak network AoI led by our algorithm follows a periodic pattern as depicted in
Fig.~\ref{pettern_peak_aoi_cyclicselction}. To start with,  
we partition the time horizon $[1,T]$ into multiple intervals as follows, 
\begin{equation}
    T_i=[1+(i-1)\underline{\mu}\Delta+\overline{\varsigma}, 1+i\underline{\mu}\Delta+\overline{\varsigma}),
\end{equation}
for $i=\{1,2,3,...\}$, i.e., $[1,T]=\bigcup\limits_{i\geq 1}T_i$. We have two cases.

\noindent \textbf{Case 1.}
($t\in T_1$). 
Since  Algorithm~\ref{peak_cycliselection_alg} selects seed $s_i=\omega_i$ at time $1+(i-1)\Delta$, our information diffusion model yields that the number of nodes updated by
$\omega_{i}$ at time ($1+\underline{\mu}\Delta+\overline{\varsigma}$) follows
\begin{equation*}
\begin{split}
1+ 2[1+\underline{\mu}\Delta+\overline{\varsigma}-1-(i-1)\Delta]=1+2(\underline{\mu}-i+1)\Delta+2\overline{\varsigma}.
   \end{split}
\end{equation*}
Due to our seed candidates fine-tuned by Algorithm \ref{peak_cycliselection_alg}, we can rewrite each seed $\omega_{\underline{\mu}-j}$ (where $j\in\{0,...,\underline{\mu}-1\}$)  as follows
\begin{equation}
 \omega_{\underline{\mu}-j}=v_{\Delta (j+1)^2+(j+1)+(2j+1)\overline{\varsigma}}. 
\end{equation}
As a result, at time ($1+\underline{\mu}\Delta+\overline{\varsigma}$),  each seed $\omega_{\underline{\mu}-j}$ 
could update the following set of head-to-tail nodes
\begin{equation}\label{peak_cyclicselction_coveredset}
\begin{split}
  \{v_{(j+1)j\Delta +(j+1)+2j\overline{\varsigma}},...,v_{(j+1)(j+2)\Delta+(j+1)+(2j+2)\overline{\varsigma}}\},
\end{split}
\end{equation}
for each $j\in\{0,...,\underline{\mu}-1\}$. 

\begin{figure}
    \centering \includegraphics[width=8.8cm]{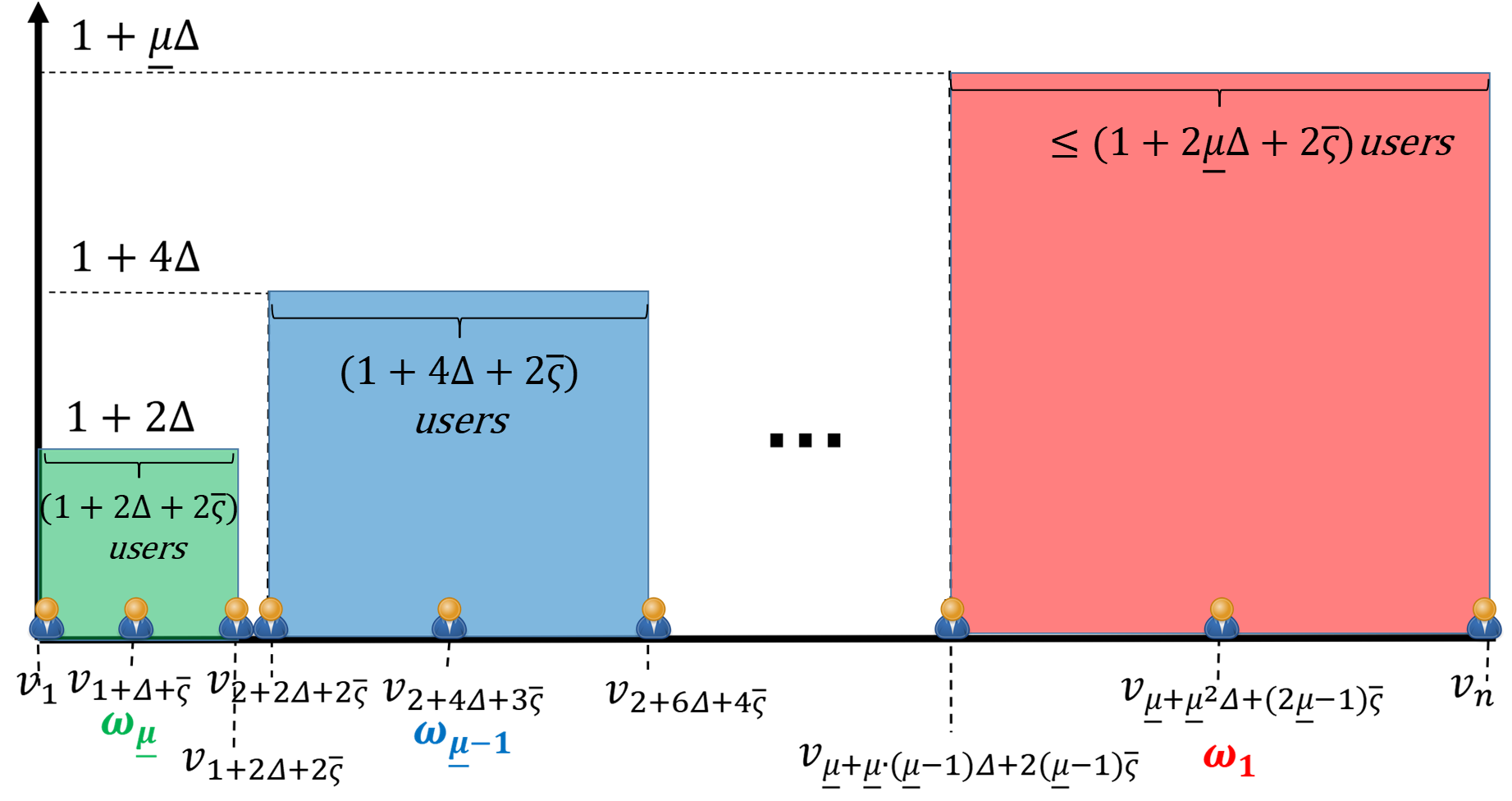}
    \caption{Configuration of seed candidates in a line-type social network, where X-axis and Y-axis indicate candidate locations and AoI at time $1+\underline{\mu}\Delta+\overline{\varsigma}$, respectively. Note that, at time $1+\underline{\mu}\Delta$, the AoI of users covered by a rectangular of the same color is no more than the height of that rectangle.}\label{peak_linetype_socialgraph_cyclic}
\end{figure}
Fig.~\ref{peak_linetype_socialgraph_cyclic} depicts the configuration of our seed candidates on a line-type social graph, which can be interpreted as follows.
First, the X-axis of Figure \ref{peak_linetype_socialgraph_cyclic} represents the relative locations among seed candidates. More specifically, each black-marked value on the X-axis indicates the position of the corresponding node, particularly, each colored mark on the X-axis indicates the position of the corresponding seed candidate that is fine-tuned by our algorithm.  The Y-axis of Figure \ref{peak_linetype_socialgraph_cyclic} represents the AoI, note that the peak AoI of each node in the same rectangular is not larger than the AoI indicated by the height of the corresponding rectangle. Consequently, we can tell that the peak AoI of the network is not larger than 
$1+\underline{\mu}\Delta$. 

Fig.~\ref{peak_linetype_socialgraph_cyclic} also illustrates that the node immediately following the right-most node updated by seed $\omega_{\underline{\mu}-j}$ is just the left-most node updated by seed $\omega_{\underline{\mu}-(j-1)}$. This implies that the family of sets (\ref{peak_cyclicselction_coveredset}) are disjoint for different $j\in\{0,...,\underline{\mu}-2\}$. By taking the union operation over sets~(\ref{peak_cyclicselction_coveredset}) over different $j\in\{0,...,\underline{\mu}-2\}$, we observe from Fig.~\ref{peak_linetype_socialgraph_cyclic} that those head-to-tail nodes in the following set are updated by seeds $\{\omega_{2},...,\omega_{\underline{\mu}}\}$ up to time $(1+\underline{\mu}\Delta+\overline{\varsigma})$:
\begin{equation}\label{set_of_nodes_covered}
    \{v_1,...,v_{(\underline{\mu}-1)\underline{\mu}\Delta+\underline{\mu}-1+(2\underline{\mu}-2)\overline{\varsigma}}\}
\end{equation}
which covers a total number $((\underline{\mu}-1)\underline{\mu}\Delta+\underline{\mu}-1+(2\underline{\mu}-2)\overline{\varsigma})$ of nodes. Accordingly, the remaining number of nodes that are not updated by  $\{\omega_{2},...,\omega_{\underline{\mu}}\}$ follows
\begin{equation}\label{the_supporting_material_toshowthecover}
\begin{split}
&n-((\underline{\mu}-1)\underline{\mu}\Delta+\underline{\mu}-1+(2\underline{\mu}-2)\overline{\varsigma})\\
    &\leq n-((\underline{\mu}-1)\underline{\mu}\Delta+\underline{\mu}-1+(2\underline{\mu}-2)\varsigma)\\
    &=1+2\underline{\mu}\Delta+2\varsigma.
    \end{split}
\end{equation}
This implies that the seed $\omega_1$ can update all the nodes outside those in (\ref{set_of_nodes_covered}).
Thus, all nodes in $V$ are updated at least once by time $(1+\underline{\mu}\Delta+\overline{\varsigma}
)$, yielding that the peak AoI in $T_1$ follows 
\begin{equation}\label{eq_case1_peak_aoi}
    A^{\rm peak}\leq A_0+1\underline{\mu}\Delta+\overline{\mu}
\end{equation}
\noindent\textbf{Case 2.}
($t\in \bigcup\limits_{i\geq 2}T_i$).
According to Step~3 of our Algorithm \ref{peak_cycliselection_alg}, there are a number $\underline{\mu}$ of new seeds selected as $(\omega_1,\omega_2,...,\omega_{\underline{\mu}})$ in each interval $T_i-\overline{\varsigma}=[1+(i-1)\underline{\mu}\Delta, 1+i\underline{\mu}\Delta)$. To distinguish those seeds selected from different intervals, we further denote $(\omega_1^i,\omega_2^i,...,\omega_{\underline{\mu}}^i)$ as those seeds that are selected in interval $T_i-\overline{\varsigma}$. 
By applying a similar analysis as in (\ref{peak_cyclicselction_coveredset}) and (\ref{the_supporting_material_toshowthecover}), it can be verified that by the end of the interval $T_i$, all nodes in $V$ are updated at least once 
by some seed candidate in $(\omega_1^i,\omega_2^i,...,\omega_{\underline{\mu}}^i)$ that are selected within $T_i-\overline{\varsigma}$. Since the start time of each interval $T_i$ is exactly the end time of the prior interval, the AoI of each node at the start time of $T_i$ is no larger than the age of the oldest information (which is 
disseminated from seed $\omega_1^{i-1}$) in the prior interval $T_{i-1}$. In other words, the initial AoI of nodes in the interval $T_i$ follows
\begin{equation*}
    A(v_i,1+(i-1)\underline{\mu}\Delta+\overline{\varsigma})\leq 1+\underline{\mu}\Delta+\overline{\varsigma}
\end{equation*}
Since $T_i$ has a duration of 
$\underline{\mu}\Delta$, the
peak AoI of each node within $T_i$ is no larger than 
\begin{equation}\label{eq_case2_peak_aoi}
1+\underline{\mu}\Delta+\overline{\varsigma}+\underline{\mu}\Delta=1+2\underline{\mu}\Delta+\overline{\varsigma}.
\end{equation}
By applying (\ref{eq_case2_peak_aoi}) and (\ref{theorem_peak_opt_eq}) back to (\ref{peak_aoi_def}), we have
\begin{equation}\label{element_peak_ratio_cyclicselection_01}
\begin{split}
    \frac{A^{\rm peak}_{\rm ALG}}{A^{\rm peak}_{\rm OPT}}&\leq \frac{1+\underline{\mu}\Delta+\overline{\varsigma}+\max\{A_0,\underline{\mu}\Delta\}}{1+\underline{\mu}\Delta+\max\{A_0+\overline{\varsigma},\underline{\xi}\Delta\}}\\
    &\leq \frac{1+2\underline{\mu}\Delta+\overline{\varsigma}}{1+\underline{\mu}\Delta+\underline{\xi}\Delta},
\end{split}
\end{equation}
where the first inequality holds by (\ref{eq_case1_peak_aoi}), (\ref{eq_case2_peak_aoi}) and  Lemma~\ref{theorem_peak_opt}.  

Particularly when $n\leq \frac{A_0^2+(\Delta+1)A_0}{\Delta}$,  we have
\begin{equation}\label{thelasteq_peak_aoi}
    A_0\geq\frac{-\Delta-1+\sqrt{\Delta^2+(2+4n)\Delta+1}}{2}= \mu\Delta\geq \underline{\mu}\Delta.
\end{equation}
By applying (\ref{thelasteq_peak_aoi}) and the fact $\xi\leq \mu$ to the first inequality of
(\ref{element_peak_ratio_cyclicselection_01}), we meet $\frac{A^{\rm peak}_{\rm ALG}}{A^{\rm peak}_{\rm OPT}}=\frac{1+\underline{\mu}\Delta+A_0+\overline{\varsigma}}{1+\underline{\mu}\Delta+A_0+\overline{\varsigma}}=1$.
\end{proof}
In view of Theorem~\ref{avg_approximation_theorem}, 
the low time complexity and approximation tell decent performances of our Algorithm~\ref{peak_cycliselection_alg}, which even achieves optimality for small networks. 
\subsection{Analysis Extension to General Social Networks}\label{subsec_general_graph}
In general networks, social connections are more complicated and AoI updates among network nodes may become more interrelated. Note that Algorithm \ref{peak_cycliselection_alg} solely seeds on the diameter of a general graph, making its corner nodes suffer from larger delay than its middle nodes. Besides, 
any social connection between nodes outside $diam(G)$ will evidently accelerate our algorithm's information diffusion and consequently ameliorate our algorithm's performance. As such, when analyzing Algorithm \ref{peak_cycliselection_alg}'s worst-case performance, it is sufficient to focus on the histogram (as defined below and shown in Figure \ref{histogram_social_graph}) that is reduced from a given general graph. Our subsequent Lemma~\ref{reduce_generagraph_lemma} corroborates this claim.

\begin{figure}[h]
    \centering    \includegraphics[width=3.5cm]{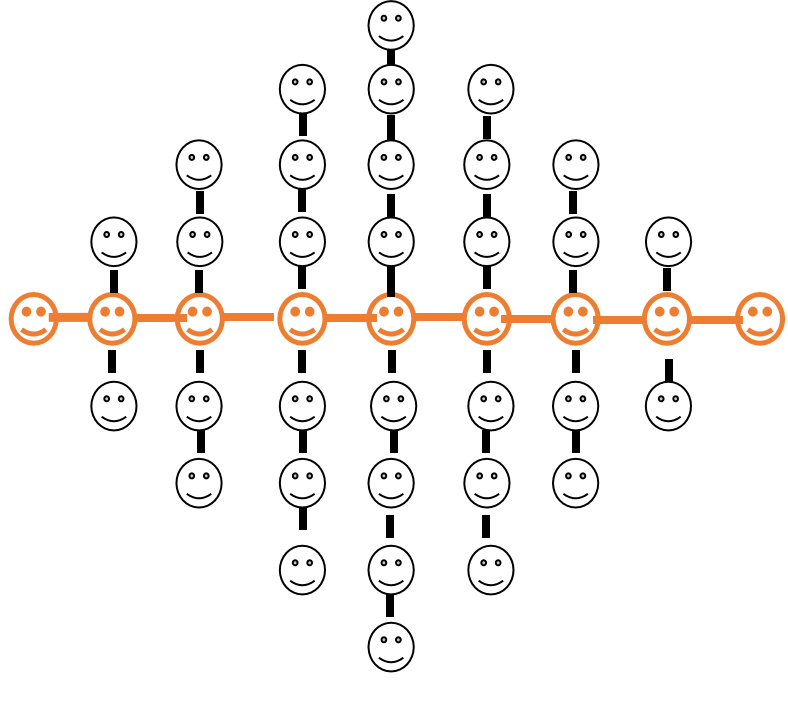}
    \caption{Illustration of a histogram network structure in which nodes on the diameter path are highlighted in orange.}   \label{histogram_social_graph}
\end{figure}
\begin{definition}[Histogram-type graph $H(V_H,E_H)$]
In a histogram-type network $H(V_H,E_H)$, each node $v\in V_H$ follows
\begin{equation}\label{eq_lemma_histogram_feature}
    \min\limits_{v'\in diam(H)} dist(v',v)\leq \min\{dist(v_{\perp},v_{\vdash}),dist(v_{\perp},v_{\dashv})\},
\end{equation}
where $v_{\vdash}$ and $v_{\dashv}$ indicate the left- and the right-most nodes on the diameter path $diam(H)$ of graph $H$, respectively, and $v_{\perp}$ denotes the closest node to $v$ on $diam(H)$, i.e., $v_{\perp}\triangleq \arg\min\limits_{v'\in diam(H)} dist(v',v)$.
\end{definition}
%
%
\begin{algorithm}[t]
\caption{Graph Reducing Approach}\label{graph_reducing_alg}
  \textbf{Data:} {A general graph $\{G\}$.}\\
  \textbf{Result:} {A reduced histogram graph $H(V_H,E_H)$ corresponding to $G$.}
  \begin{algorithmic}[1]
\STATE Find an arbitrary diameter path $diam(G)$ (if multiple exist) of the given graph $G$. 
\STATE Let $V(diam(G))$ summarize those nodes that appear on $diam(G)$, and let $W \triangleq V- V(diam(G))$ summarize those nodes that do not appear on $diam(G)$.
\STATE Initialization. $V_H=V(diam(G)),E_H=E(diam(G))$.

    \FOR{$v_i\in W$}
    \STATE Find $v^{*}_i\triangleq \arg\min\limits_{v\in V_H}\{dist(v,v_i)\}$ (if multiple $v^*_i$ exist, higher priority is given to a node on $diam(G)$).
\STATE Update $V_H= V_H+\{v_i\}$ and $E_H= E_H+\{(v_i,v^*_i)\}$.
    \ENDFOR
    \end{algorithmic}
\end{algorithm}
The following lemma gives insights for our subsequent approximation analysis. 
\begin{lemma}\label{reduce_generagraph_lemma}
The approximation performance of Algorithm~\ref{peak_cycliselection_alg} in a general social network can be effectively evaluated using a histogram-type social network with the same diameter. 
\end{lemma}

In Algorithm~\ref{graph_reducing_alg}, we reduce any given general graph $G$ to its corresponding histogram graph $H(G)$ with the same diameter, while having no bearing on our approximation analysis. At the high level, our Algorithm \ref{graph_reducing_alg} first copies $diam(G)$ to $H(G)$; then, for each node $v_i$ of $G$ that lies outside $diam(G)$, Algorithm \ref{graph_reducing_alg} finds in the current node set $V_H$ the node $v^*_i$ that is closest to $v_i$ under graph $G$, and adds node $v_i$ and edge $(v_i,v_i^*)$ to $V_H$ and $E_H$ of graph $H(G)$, respectively. 

Then, it becomes tractable to analyze the performance of our Algorithm \ref{graph_reducing_alg} for $H(G)$.  Evidently, the peak AoI produced by  Algorithm \ref{graph_reducing_alg} in a general graph will not exceed $A_0+diam(G)$. Considering the information diffusion among nodes on the diameter path in a general graph, we can apply the right-hand-side of (\ref{theorem_peak_opt_eq}) in Lemma \ref{theorem_peak_opt} as a lower bound on the optimal solution for a general graph. Thus, we obtain the approximation guarantee for Algorithm~\ref{peak_cycliselection_alg} in general graphs, as stated in Theorem \ref{general_peak_approximation}.
\begin{theorem}\label{general_peak_approximation}
For the peak AoI minimization problem (\ref{problem_peak_aoi}), Algorithm~\ref{peak_cycliselection_alg} guarantee an approximation $\frac{diam(G)+A_0}{\max\{A_0+\overline{\varsigma},\left \lfloor {\xi} \right \rfloor\Delta\}+1+\underline{\mu}\Delta}$ of the optimum.
\end{theorem}

 Our Theorem~\ref{problem_peak_aoi} above shows that in the worst-case scenario, Algorithm~\ref{peak_cycliselection_alg} approximates the optimum about the square root of the network diameter due to constraint (\ref{enough_updates_constraint}) on $\underline{\mu}$. In other words, the peak AoI by Algorithm~\ref{peak_cycliselection_alg} is guaranteed to be no worse than $O(\sqrt{diam(G)})$ of the optimum in general. The worst-case performance remains robust against variations in other system parameters (such as $A_0$ and $\Delta$).
In the following, we present our solution to the average AoI minimization problem (\ref{problem_average_aoi}). 
\section{Average AoI Minimization Algorithms}\label{sec_average_aoi}
Unlike the peak AoI problem (\ref{problem_peak_aoi}), 
 optimizing average AoI is more involved as it explicitly takes into account the AoI dynamics of every user.
Despite our closed-form expression (\ref{avgaoi_formulation}) of the average AoI objective in Theorem~\ref{avgaoi_formulation_lemma}, finding a feasible solution to our problem is still challenging due to its NP-hard combinatorial nature. 

Thereby, 
we equivalently transform the objective in (\ref{avgaoi_formulation}) to the one in the following  Proposition~\ref{corollary_average_Aoi_over_the network}. With this transformation, we develop efficient approximation algorithms that provide near-optimal solutions.
%
\begin{proposition}\label{corollary_average_Aoi_over_the network}
Given the sequence $S_k=(s_1,...,s_k)$ of dynamically selected seeds, 

\begin{equation}\label{corollary_average_Aoi_over_the network_eq}
\begin{split}
   A^{\rm avg}=\eta+\frac{1}{n}\sum\limits_{v_i\in V}\Big\{&\underbrace{2A_0\cdot [\Delta  i_1+dist(s_{i_1},v_i)]}_{\rm denoted\;as\;term\;1}\\  &+\underbrace{2\Delta\cdot\sum\limits_{j=1}^{k_i}[(i_j-i_{j-1})\cdot dist(s_{i_j},v_i)]}_{\rm denoted\;as\;term\;2}\\
   &+\underbrace{\Delta^2\sum\limits_{j=1}^{k_i}(i_j-i_{j-1})^2}_{\rm denoted\;as\;term\;3}\Big\},
   \end{split}
\end{equation}
where $i_0=1$ initially, indices in set $\{i_1,...i_{k_i}\}$ 
are obtained by running Algorithm~\ref{DiscontinuityPoint} on $S_k$, and 
\begin{equation}\label{eta_formulation}
\begin{split}
    \eta
    = \frac{1}{2T}[&2A_0-2\Delta A_0+T^2+2\Delta T+\Delta^2-2\Delta\\
    &+2(\Delta-T\Delta-\Delta^2)k+\Delta^2 k^2].
\end{split}
\end{equation}
\end{proposition}
Note that $\eta$ in (\ref{eta_formulation}) is a constant. Then, we focus on the other three terms in (\ref{corollary_average_Aoi_over_the network_eq}) and aim to rigorously lower and upper bounding each of the three terms in (\ref{corollary_average_Aoi_over_the network_eq}) for further approximation analysis, respectively. Consequently, we have the following lemma.
\begin{lemma}\label{lemma_boundingterm1}
 Given $S_k=(s_1,...,s_k)$, the followings hold:
\begin{small}
\begin{equation*}
     A_0\Delta\leq A_0\cdot [i_1\Delta +dist(s_{i_1},v_i)]\leq A_0\Delta+A_0\cdot dist(s_{1},v_1), 
\end{equation*}
 \begin{equation*}
 \begin{split}
    {\color{black} \Delta\cdot\sum\limits_{j=1}^{k}\frac{ dist(s_{j},v_i)}{\beta}}&\leq\Delta\cdot\sum\limits_{j=1}^{k_i}[(i_j-i_{j-1})\cdot dist(s_{i_j},v_i)]\\
    &\leq \Delta\cdot\sum\limits_{j=1}^{k}dist(s_{j},v_i),
    \end{split}
 \end{equation*}
 \begin{equation*}
  \frac{(k-1)^2\Delta^2}{2k_i}\leq\frac{1}{2}\Delta^2\sum\limits_{j=1}^{k_i}(i_j-i_{j-1})^2\leq{\color{black} \frac{(k-1)^2\Delta^2}{2}},
 \end{equation*}
 \end{small}
 where $\beta$ indicates the longest distance between a seed in $S_k$ and any node in $V$ and is less than $|diam(G)|$.
 \end{lemma}
\begin{proof}
Given $S_k$ of dynamically selected seeds, we discuss the following three cases to prove the three families of inequalities in  Lemma~\ref{lemma_boundingterm1}, respectively. 

\noindent\textbf{Case 1.} We first prove 
\begin{equation*}
    A_0\Delta\leq A_0\cdot [i_1\Delta +dist(s_{i_1},v_i)]\leq A_0\Delta+A_0\cdot dist(s_{1},v_1)
\end{equation*}

Due to the fact that $t_{i_1}\geq t_1= 1$, one can easily check
\begin{equation}
    A_0\Delta\leq A_0\cdot [i_1\Delta +dist(s_{i_1},v_i)].
\end{equation}
Since
$t_{i_1}=(i_1-1)\Delta+1$, we have the following 
\begin{equation}\label{boundding_term1_upper_eq1}
\begin{split}
    &A_0\cdot [i_1\Delta +dist(s_{i_1},v_i)]\\
    &=A_0\cdot (\Delta-1)+A_0\cdot[t_{i_1}+dist(s_{i_1},v_i)]
\end{split}
\end{equation}
Since seed $s_{i_1}$ leads to the first discontinuity point of node $v_i$ at time  $t_{i_1}+dist(s_{i_1},v_i)$, we know
\begin{equation}\label{boundding_term1_upper_eq2}
    t_{i_1}+dist(s_{i_1},v_i)\leq t_1+dist(s_1,v_i).
\end{equation}
By substituting  (\ref{boundding_term1_upper_eq2}) and $t_1=1$ in   (\ref{boundding_term1_upper_eq1}), we get,

\begin{equation}
    A_0\cdot [i_1\Delta +dist(s_{i_1},v_i)]\leq A_0\Delta+A_0\cdot dist(s_1,v_1).
\end{equation}

\noindent\textbf{Case 2.} We now prove that 
\begin{small}
 \begin{equation*}
 \begin{split} \Delta\cdot\sum\limits_{j=1}^{k}\frac{ dist(s_{j},v_i)}{\beta}&\leq\Delta\cdot\sum\limits_{j=1}^{k_i}[(i_j-i_{j-1})\cdot dist(s_{i_j},v_i)]\\
 &\leq \Delta\cdot\sum\limits_{j=1}^{k}dist(s_{j},v_i),
     \end{split}
 \end{equation*}
 \end{small}
 where $\beta$ tells the longest distance between a seed in $S_k$ and any node in $V$ and is less than $diam(G)$.

For any seed $s_x$ that is selected in between $t_{i_j}$ and $t_{i_{j-1}}$ but does not lead to a discontinuity point of $v_i$, we have, on one hand, 
\begin{equation}
\begin{split}
   & t_x+dist(s_x,v_i)\geq t_{i_j}+dist(s_{i_j},v_i)\\
   \Leftrightarrow \;&  0<t_{i_j}-t_x\leq dist (s_x,v_i)-dist(s_{i_j},v_i), 
\end{split}
\end{equation}
which further implies the following 
    \begin{equation}\label{part_1_for_derivation_avg_eq01}
        (i_j-i_{j-1})dist(s_{i_j},v_i)\leq \sum_{x=i_{j-1}+1}^{i_j}dist(s_{x},v_i).
    \end{equation}
On the other hand, we have, by $\frac{dist(s_x,v_i)}{dist(s_{i_j},v_i)}\leq \beta$, that
\begin{equation}\label{part_1_for_derivation_avg_eq02}
    (i_j-i_{j-1})dist(s_{i_j},v_i)\geq \sum_{x=i_{j-1}+1}^{i_j}\frac{dist(s_x,v_i)}{\beta}.
\end{equation}
 By summing up (\ref{part_1_for_derivation_avg_eq01}) over all $j\in\{1,...,k_i\}$, we get 
 \begin{equation}
\Delta\cdot\sum\limits_{j=1}^{k_i}[(i_j-i_{j-1})\cdot dist(s_{i_j},v_i)]\leq \Delta\cdot\sum\limits_{j=1}^{k}dist(s_{j},v_i).  
 \end{equation}
By summing up (\ref{part_1_for_derivation_avg_eq02}) over all $j$, we have 
\begin{equation}
    \Delta\cdot\sum\limits_{j=1}^{k}\frac{ dist(s_{j},v_i)}{\beta}\leq\Delta\cdot\sum\limits_{j=1}^{k_i}[(i_j-i_{j-1})\cdot dist(s_{i_j},v_i)].
\end{equation}

\noindent\textbf{Case 3. } We prove in this case that  
\begin{equation}
\begin{split}
    \frac{(k-1)^2\Delta^2}{2k_i}\leq\frac{1}{2}\Delta^2\sum\limits_{j=1}^{k_i}(i_j-i_{j-1})^2\leq\frac{(k-1)^2\Delta^2}{2}
    \end{split}
\end{equation}

Since $i_j-i_{j-1}\geq 1$ holds for each $j\in\{1,..,k_i\}$, we have, by the Cauchy-Schwarz Inequality, that
\begin{equation}\label{bounding_term3_material01}
     [\sum\limits_{j=1}^{k_i}(i_j-i_{j-1})^2]\cdot [\sum\limits_{j=1}^{k_i}1^2]\geq [\sum\limits_{j=1}^{k_i}1\cdot(i_j-i_{j-1})]^2.
\end{equation}
As $\sum\limits_{j=1}^{k_i}(i_j-i_{j-1})=i_{k_i}-i_0=k-1$, Inequality (\ref{bounding_term3_material01}) implies
\begin{equation}
    \sum\limits_{j=1}^{k_i}(i_j-i_{j-1})^2\geq\frac{(k-1)^2}{k_i}.
\end{equation}
Since $x_1^2+x_2^2+\cdots+x_k^2<(x_1+x_2+\cdots+x_k)^2$ holds for every $k$ positive numbers $\{x_1,...,x_k\}$, we get 
\begin{equation}
\begin{split}
    \sum\limits_{j=1}^{k_i}(i_j-i_{j-1})^2&\leq [\sum\limits_{j=1}^{k_i}(i_j-i_{j-1})]^2\\
    &=(i_{k_i}-i_0)^2=(k-1)^2.
\end{split}
\end{equation}
This completes the proof.
\end{proof}
By applying Lemma~\ref{lemma_boundingterm1} back to (\ref{corollary_average_Aoi_over_the network_eq}) of Proposition~\ref{corollary_average_Aoi_over_the network}, we achieve a rigorous two-sided bound on the objective $A^{\rm avg}$  as in the following theorem.
\begin{theorem}\label{avg_lowerbound_ratio}
Given the sequence $S_k=(s_1,...,s_k)$ of dynamically selected seeds, the average AoI objective (\ref{corollary_average_Aoi_over_the network_eq}) can be lower and upper bounded as:
\begin{equation}\label{lowerbound_eq_avg}
\begin{split}
      A^{\rm avg}&\geq 2A_0\Delta+\frac{n(k-1)^2\Delta^2}{\sum\limits_{v_i\in V}k_i}\\
      &\quad +\frac{2\Delta}{n\beta}\sum\limits_{v_i\in V}\sum\limits_{s_j\in S_k}dist(s_j,v_i)+\eta,
\end{split}
\end{equation}
and
\begin{equation}\label{upperbound_eq_avg_second}
\begin{split}
    A^{\rm avg}&\leq 
     2A_0\Delta+(k-1)^2\Delta^2+\frac{2A_0}{n}\sum\limits_{v_i\in V}dist(s_1,v_i)\\
     &\quad +\frac{2\Delta}{n}\sum_{v_i\in V}\sum_{s_j\in S_k}dist(s_j,v_i) +\eta.
\end{split}
\end{equation}
\end{theorem}
\begin{proof}
To start with, one can easily verify that (\ref{lowerbound_eq_avg}) holds readily by applying the second inequalities in the three inequality families of  Lemma \ref{lemma_boundingterm1} to Proposition~\ref{corollary_average_Aoi_over_the network}.

To prove (\ref{upperbound_eq_avg_second}), we apply the first inequalities in the three inequality families of Lemma \ref{lemma_boundingterm1} to Proposition~\ref{corollary_average_Aoi_over_the network}. Consequently, we have  
\begin{equation}\label{temp_lowerbound_eq_avg}
\begin{split}
A^{\rm avg}\geq& 2A_0\Delta+\sum\limits_{v_i\in V}\frac{(k-1)^2\Delta^2}{nk_i}\\
&+\frac{2\Delta}{n\beta}\sum\limits_{v_i\in V}\sum\limits_{s_j\in S_k}dist(s_j,v_i)+\eta.
\end{split}
\end{equation}
The well-known AM-GM inequality \cite{steele2004cauchy} admits
\begin{equation}\label{temp03_lowerbound_eq_avg}
\frac{n}{\sum\limits_{v_i\in V}\frac{1}{k_i}}\leq \frac{\sum\limits_{v_i\in V}k_i}{n}.
\end{equation}
By reorganizing (\ref{temp03_lowerbound_eq_avg}), we have
\begin{equation}\label{temp02_lowerbound_eq_avg}
\sum\limits_{v_i\in V}\frac{1}{k_i}\geq \frac{n^2}{\sum\limits_{v_i\in V}k_i}.
\end{equation}
By plugging (\ref{temp02_lowerbound_eq_avg}) back to (\ref{temp_lowerbound_eq_avg}), (\ref{upperbound_eq_avg_second}) holds readily.
\end{proof}

In the lower bound (\ref{lowerbound_eq_avg}) and upper bound (\ref{upperbound_eq_avg_second}) above, we observe a common summation term $\sum\limits_{v_i\in V}\sum\limits_{s_j\in S_k}dist(s_j,v_i)$. This enables us to link our problem to the following Problem~\ref{ref_problem_01}, which is significantly simplified and can be solved efficiently. 
\begin{problem}[Sum-distance minimization problem]\label{ref_problem_01}
Given a graph $G(V,E)$, the objective is to select a subset $S\subseteq V$ with size $k$ to minimize the sum distance $\sum\limits_{s_j\in S_k}\sum\limits_{v_i\in V}dist(s_j,v_i)$.
\end{problem}
\begin{algorithm}[t]
\caption{\textsc{$k$-minisum} for average AoI minimization}\label{ref_problem_01_alg}
\textbf{Data:} { $G=(V,E)$, $k=\frac{T}{\Delta}$.}\\
\textbf{Result:} {$S_k$.} 
\begin{algorithmic}[1]
\STATE Initialize the distance matrix $\boldsymbol{M_G}=(dist(v_i,v_j))_{n\times n}$ of graph $G$ and unit vector $\boldsymbol{M_1}=(1)_{n\times 1}$.
\STATE Compute $\boldsymbol{M_{S}}\leftarrow \boldsymbol{M_G}\cdot \boldsymbol{M_1}$.
\STATE Find in $\boldsymbol{M_{S}}$ indices corresponding to the first $k$ smallest entries, resulting in sequence $\Theta$.
\STATE Select in set $V$ the $k$ nodes indexed by indices of $\Theta$ sequentially, forming $S_k$.
\end{algorithmic}
\end{algorithm} 

This intriguing
connection discloses a tractable way for our algorithm design and approximation analysis. In Algorithm~\ref{ref_problem_01_alg}, we present our approach to minimizing $A^{\rm avg}$. In general, Algorithm~\ref{ref_problem_01_alg} first constructs the distance matrix $\boldsymbol{M_G}=(dist(v_i,v_j))_{n\times n}$ that contains all pairwise shortest distances among nodes in the social network \cite{chan2012all}; then, according to the sum distance of a node over all the other nodes, Algorithm~\ref{ref_problem_01_alg} seeds sequentially in decreasing order of nodes' sum distances. 
Now, we further look at Algorithm~\ref{ref_problem_01_alg}' average AoI in general social networks. 
\begin{theorem}\label{theorem_average_guarantee}
For the average AoI minimization problem (\ref{problem_average_aoi}) in general social networks, Algorithm~\ref{ref_problem_01_alg} runs in $O(mn\cdot\frac{\log\log n}{\log n}+n^2\cdot \frac{\log^2\log n}{\log n})$-time, which is strictly less than $O(mn)$, and outputs a multi-stage seeding $S_k=(s_1,...,s_k)$ that generally achieves an approximation guarantee of no worse than $\max\bigg\{\beta+\frac{A_0\beta}{k\Delta},
\frac{2A_0\Delta+\eta+(k-1)^2\Delta^2}{2A_0\Delta+\eta+\frac{(k-1)^2\Delta^2}{k}}\bigg\}$. 
\end{theorem}
\begin{proof}
The running time of Algorithm~\ref{ref_problem_01_alg} is dominated by its Step 1 in returning all-pairs-shortest-path (APSP) distances of a given graph $G(V,E)$. Thanks to the state-of-art result for APSP \cite{chan2012all}, our Algorithm~\ref{ref_problem_01_alg} totally runs in $O(mn\cdot\frac{\log\log n}{\log n}+n^2\cdot \frac{\log^2\log n}{\log n})$-time.

Next, we prove the approximation guaranteed by Algorithm~\ref{ref_problem_01_alg}.
Denote $S_k=(s_1,...,s_k)$ and $S_k^*=(s_1^*,...,s_k^*)$ as the $k$ seeds that are dynamically-selected by Algorithm~\ref{ref_problem_01_alg} and an optimal solution, respectively. Recall that $A^{\rm avg}_{\rm OPT}$ and $A^{\rm avg}_{\rm ALG}$ denote the average AoI results of an optimal solution and our algorithm, respectively. According to Theorem~\ref{avg_lowerbound_ratio}, we get
\begin{equation}\label{avg_ratio_inequality01}
\begin{split}
   A^{\rm avg}_{\rm OPT}\geq &\underbrace{2A_0\Delta+\eta+\frac{n(k-1)^2\Delta^2}{\sum\limits_{v_i\in V}k_i^*}}_{{\rm denoted,\;for\;short,\;as\;}A^{\rm avg}_{\rm OPT}|_1}\\
   &+\underbrace{\frac{2\Delta}{n\beta}\sum\limits_{v_i\in V}\sum\limits_{s_j^*\in S_k^*}dist(s_j^*,v_i)}_{{\rm denoted,\;for\;short,\;as\;}A^{\rm avg}_{\rm OPT}|_2},
\end{split}
\end{equation}
\begin{equation}\label{avg_ratio_inequality02}
\begin{split}
    & A^{\rm avg}_{\rm ALG}\\
    &\leq 
     \underbrace{2A_0\Delta+(k-1)^2\Delta^2+\eta}_{{\rm denoted,\;for\;short,\;as\;}A^{\rm avg}_{\rm ALG}|_1}\\
     &\quad +\underbrace{\frac{2A_0}{n}\sum\limits_{v_i\in V}dist(s_1,v_i)+\frac{2\Delta}{n}\sum_{v_i\in V}\sum_{s_j\in S_k}dist(s_j,v_i)}_{{\rm denoted,\;for\;short,\;as\;}A^{\rm avg}_{\rm ALG}|_2}.
     \end{split}
\end{equation}
By applying (\ref{avg_ratio_inequality01}) and (\ref{avg_ratio_inequality02}) into (\ref{avg_aoi_def}), we further obtain
\begin{equation}\label{avg_ratio_inequality03}
\begin{split}
\frac{A^{\rm avg}_{\rm ALG}}{A^{\rm avg}_{\rm OPT}}&\leq \frac{A^{\rm avg}_{\rm ALG}|_1+A^{\rm avg}_{\rm ALG}|_2}{A^{\rm avg}_{\rm OPT}|_1+A^{\rm avg}_{\rm OPT}|_2}\\
&\leq \max\{\frac{A^{\rm avg}_{\rm ALG}|_1}{A^{\rm avg}_{\rm OPT}|_1},
\frac{A^{\rm avg}_{\rm ALG}|_2}{A^{\rm avg}_{\rm OPT}|_2}
\},
\end{split}
\end{equation}
where the first inequality holds by 
 (\ref{avg_ratio_inequality01}) and (\ref{avg_ratio_inequality02}), and the second inequality is due to the fact that $\frac{x_1+x_2}{y_1+y_2}\leq \max\{\frac{x_1}{y_1},\frac{x_2}{y_2}\}$ holds for any four positive numbers $x_1,y_1,x_2,y_2$. On one hand, note that the seed set $S_k$ selected by our Algorithm~\ref{ref_problem_01_alg} also minimizes Problem~\ref{ref_problem_01}, i.e., 
\begin{equation}\label{avg_ratio_inequality04}
\sum_{v_i\in V}\sum_{s_j\in S_k}dist(s_j,v_i)\leq \sum_{v_i\in V}\sum_{s_j^*\in S_k^*}dist(s_j^*,v_i).
\end{equation}
Due to Step  3 of our Algorithm~\ref{ref_problem_01_alg}, we have on the other hand that
\begin{equation} \label{avg_ratio_inequality05}
\sum\limits_{s_j\in S_k}\sum\limits_{v_i\in V}dist(s_j,v_i)\geq k\cdot \sum\limits_{v_i\in V}dist(s_1,v_i).
\end{equation}
Considering the fact that $\sum\limits_{v_i\in V}dist(s_1,v_i)\leq \sum\limits_{v_i\in V}dist(s_j,v_i)$ holds for each $s_j\in S_k$, the following which 
stems from (\ref{avg_ratio_inequality01}) and (\ref{avg_ratio_inequality02}) holds:

\begin{equation}\label{avg_ratio_inequality06}
\begin{split}
    \frac{A^{\rm avg}_{\rm ALG}|_2}{A^{\rm avg}_{\rm OPT}|_2}
      &\leq \frac{(\frac{2A_0}{nk}+\frac{2\Delta}{n})\cdot\sum\limits_{v_i\in V}\sum\limits_{s_j\in S_k}dist(s_j,v_i)}{\frac{2\Delta}{n\beta}\sum\limits_{v_i\in V}\sum\limits_{s_j^*\in S_k^*}dist(s_j^*,v_i)}\\
        &\leq \frac{(\frac{2A_0}{nk}+\frac{2\Delta}{n})\cdot\sum\limits_{v_i\in V}\sum\limits_{s_j^*\in S_k^*}dist(s_j^*,v_i)}{\frac{2\Delta}{n\beta}\sum\limits_{v_i\in V}\sum\limits_{s_j^*\in S_k^*}dist(s_j^*,v_i)}\\
        &= \frac{A_0\beta}{k\Delta}+\beta,
\end{split}
\end{equation}
in which the first and the second inequalities are due to
 (\ref{avg_ratio_inequality05}) and (\ref{avg_ratio_inequality04}), respectively. 
 
 Since $\sum\limits_{v_i\in V}k^*_i\leq \sum\limits_{v_i\in V}k=nk$, we thus have 
 \begin{equation}\label{avg_ratio_inequalitylast}
 \begin{split}
     \frac{A^{\rm avg}_{\rm ALG}|_1}{A^{\rm avg}_{\rm OPT}|_1}&\leq \frac{2A_0\Delta+(k-1)^2\Delta^2+\eta}{2A_0\Delta+\eta+\frac{n(k-1)^2\Delta^2}{\sum\limits_{v_i\in V}k_i^*}}\\
     &\leq \frac{2A_0\Delta+\eta+(k-1)^2\Delta^2}{2A_0\Delta+\eta+\frac{(k-1)^2\Delta^2}{k}}
\end{split}
 \end{equation}

By applying (\ref{avg_ratio_inequality06}) and (\ref{avg_ratio_inequalitylast}) to  (\ref{avg_ratio_inequality03}), the proof completes.
\end{proof}

In Theorem~\ref{theorem_average_guarantee} of our average AoI result, the parameter $\beta$ in the first item of the approximation guarantee represents the furthest distance from a selected seed to any other node in the social network. This $\beta$ is influenced by both the size and structure of the network and the seeding budget $k$. The second item in the approximation guarantee corresponds to both the seeding budget $k$ and the parameter $\eta$, which is given in (\ref{eta_formulation}). This $\eta$ is related to the time horizon $T$ and significantly mitigates the effect of $k$ in the second item of the approximation guarantee, as it appears in both the numerator and denominator of the guarantee's second item. On one hand, when the network is small and the seeding budget $k$ is relatively large enough, the worst-case performance guarantee above in Theorem~\ref{theorem_average_guarantee} will be dominated by the second item. This tells that expanding the time horizon $T$ will narrow the gap between our algorithm's solution and the optimum, yielding our solution's better approximation performance. On the other hand, when the network is particularly large with a small seeding budget, the worst-case performance guarantee is dominated by its first item, which approximates the network diameter divided by the seed budget, i.e., $\frac{|diam|}{k}$. This tells our algorithm's decent approximation performance even for large networks.

%
%
%
\section{Experiments}\label{sec_experiments}
\begin{figure}[t]
  \centering
    \includegraphics[width=8cm]{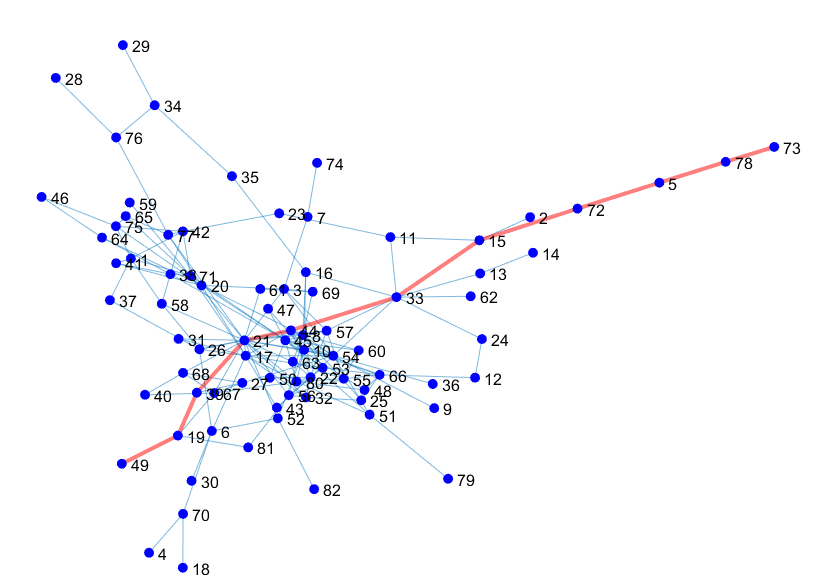}
    \caption{Social network visualization. } \label{our_graph_inmplementation_}
\end{figure}
\begin{figure}[h]
    \centering
\includegraphics[width=8cm]{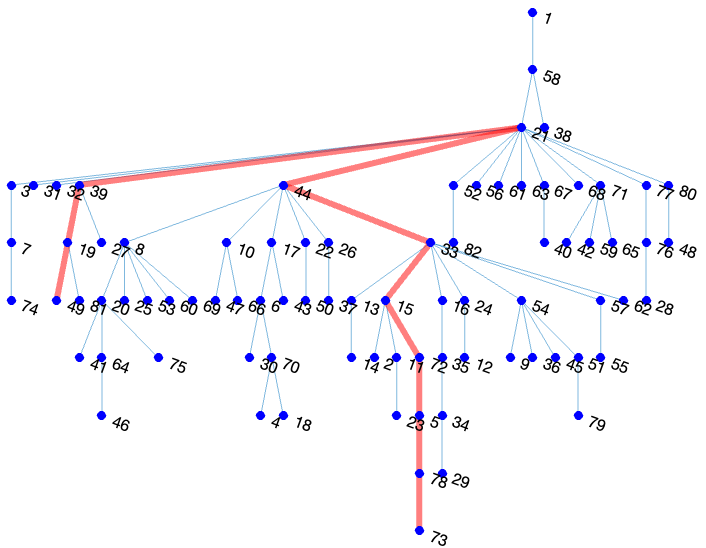}
    \caption{Histogram graph $H(G)$ reduced by Algorithm \ref{graph_reducing_alg} on $G$ in Figure \ref{our_graph_inmplementation_}.}
    \label{reduced_histogram}
\end{figure}
In this section, we empirically evaluate our algorithms by experiments on a 
real data set of Facebook social circles \cite{leskovec2012learning,realsocialgraph}.  
In the following implementation, we consider the typical 100 users with their IDs  $\{1,2,...,100\}$ in the original data set \cite{leskovec2012learning,realsocialgraph}, from which we remove some isolated users that have no bearing on our experimental results. Consequently, our social network  $G(V,E)$ is visualized in Fig.~\ref{our_graph_inmplementation_}, where the diameter path $diam(G)$ is highlighted in thick red. {\color{black} All of our implementations are conducted in Matlab R2022A. 

Since it is NP-hard to find optimal $A^{\rm peak}_{\rm OPT}$ and $A^{\rm avg}_{\rm OPT}$, we are inspired by Theorems~\ref{lemma_aoi_peak_v_i} and \ref{eq_peak_vi} to 
apply the following lower bounds ${\rm LB_{peak}}$ and ${\rm LB_
{avg}}$
 as our algorithms' comparison benchmarks, respectively:
\begin{equation}
\begin{split}
    {\rm LB}_{\rm peak}&=A_{0}+1+\frac{1}{n(n-1)}\sum\limits_{v_j\in V}\sum\limits_{v_i\in V-\{v_j\}} dist(v_j,v_i)\\
    &\leq A_{0}+1+\max\limits_{v_i,v_j\in V}\{dist(v_j,v_i)\}\\
    &\overset{(\ref{eq_peak_vi})}{\leq} A^{\rm peak}_{\rm OPT}.
\end{split}
\end{equation}
\begin{equation}\label{avg_implementation_01}
\begin{split}
    {\rm LB_{avg}} &=  2A_0\Delta+\eta+\frac{(k-1)^2\Delta^2}{k}+\frac{2\Delta
    P_1(G)}{n \cdot |diam(G)|}\\
    &\overset{(\ref{lowerbound_eq_avg})}{\leq}A^{\rm avg}_{\rm OPT},
\end{split}
\end{equation}
where $\eta$ refers to (\ref{eta_formulation}) and 
$$P_1(G)\triangleq\min\limits_{V'\subseteq V,|V'|=k} \sum\limits_{v_i\in V}\sum\limits_{v_j\in V'}dist(v_i,v_j)$$
 indicates the optimum of the sum-distance minimization problem in Problem 1 on graph $G$. It is important to note that both ${\rm LB_{peak}}$ and ${\rm LB_{avg}}$ are independent of how an optimal solution seeds, since they are general lower bounds on our peak and average AoI objectives. 

For peak AoI minimization, we implement our Algorithm~\ref{peak_cycliselection_alg} with different settings of the seeding interval $\Delta = 1$ and $\Delta = 2$, respectively, and vary the time horizon till $T=80$ time slots. To start with, we present, in Figure \ref{reduced_histogram}, the histogram graph that our Algorithm \ref{graph_reducing_alg} reduces from the input social network in Figure \ref{our_graph_inmplementation_}.
Our experimental results are presented in Fig.~\ref{peak_result_fig}.  Recall  (\ref{eq01_prop2}) in our Proposition \ref{prop_def_three_values} that $\underline{\mu}$ decreases as $\Delta$ increases, i.e., our Algorithm \ref{peak_cycliselection_alg} selects fewer candidates on $diam(G)$ as seeding interval $\Delta$ enlarges. This implies that larger $\Delta$ will induce a larger peak AoI approximation in our Algorithm \ref{peak_cycliselection_alg}, as illustrated in Fig.~\ref{peak_result_fig}. The reason for the early drop in Fig.~\ref{peak_result_fig} is two-fold: first, our benchmark ${\rm LB_{peak}}$ is chosen as a constant value that is smaller than $A^{\rm peak}_{\rm OPT}$; second, our algorithm could achieve better performance when the overall seed number $k$ increases in an early stage. As $k$ exceeds a threshold of 20 here, our empirical peak AoI ratio stabilizes at approximately 1.6 for $\Delta =1$ and 1.4 for $\Delta =2$. This validates our theoretical finding as aforementioned in Section \ref{sec_peak_aoi}
that Algorithm \ref{peak_cycliselection_alg} results in a periodic pattern in its peak AoI dynamics. 

To further evaluate our algorithm, in Fig.~\ref{peak_result_fig}, we also compare our Algorithm~3 with a baseline greedy algorithm that always seeds the node with the highest AoI in each seeding round. Our experimental results here demonstrate that Algorithm~\ref{peak_cycliselection_alg} consistently outperforms the greedy approach, as the latter fails to exploit the network connectivity to effectively reduce nodes' AoI in the long run. 

\begin{figure}[t]
    \centering
    \includegraphics[width=8.8cm]{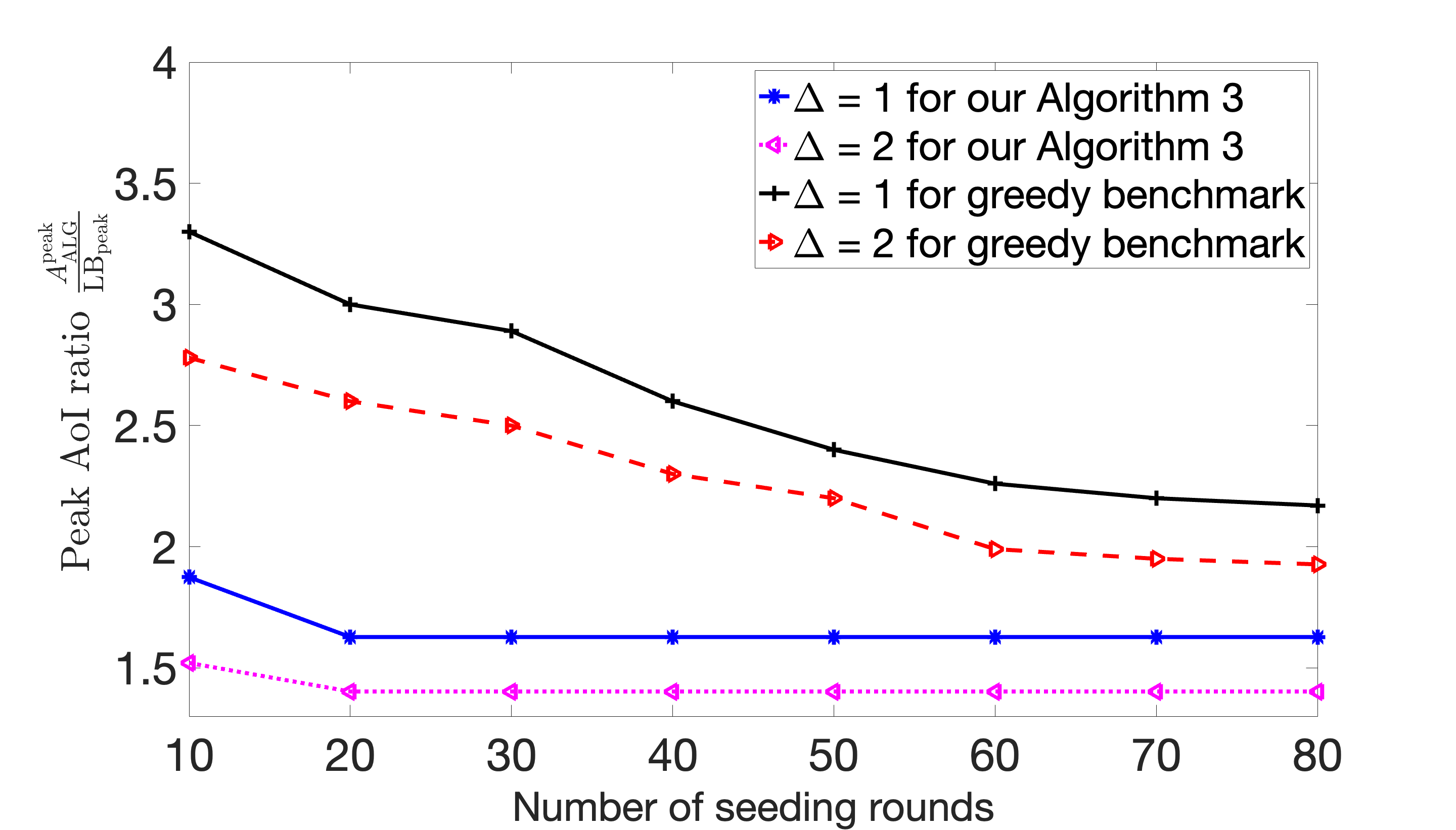}
    \caption{Peak AoI performance ratio $\frac{A^{\rm peak}_{\rm ALG}}{{\rm LB}_{\rm peak}}$.}
    \label{peak_result_fig}
\end{figure}

\begin{figure}[t]
    \centering
    \includegraphics[width=9cm]{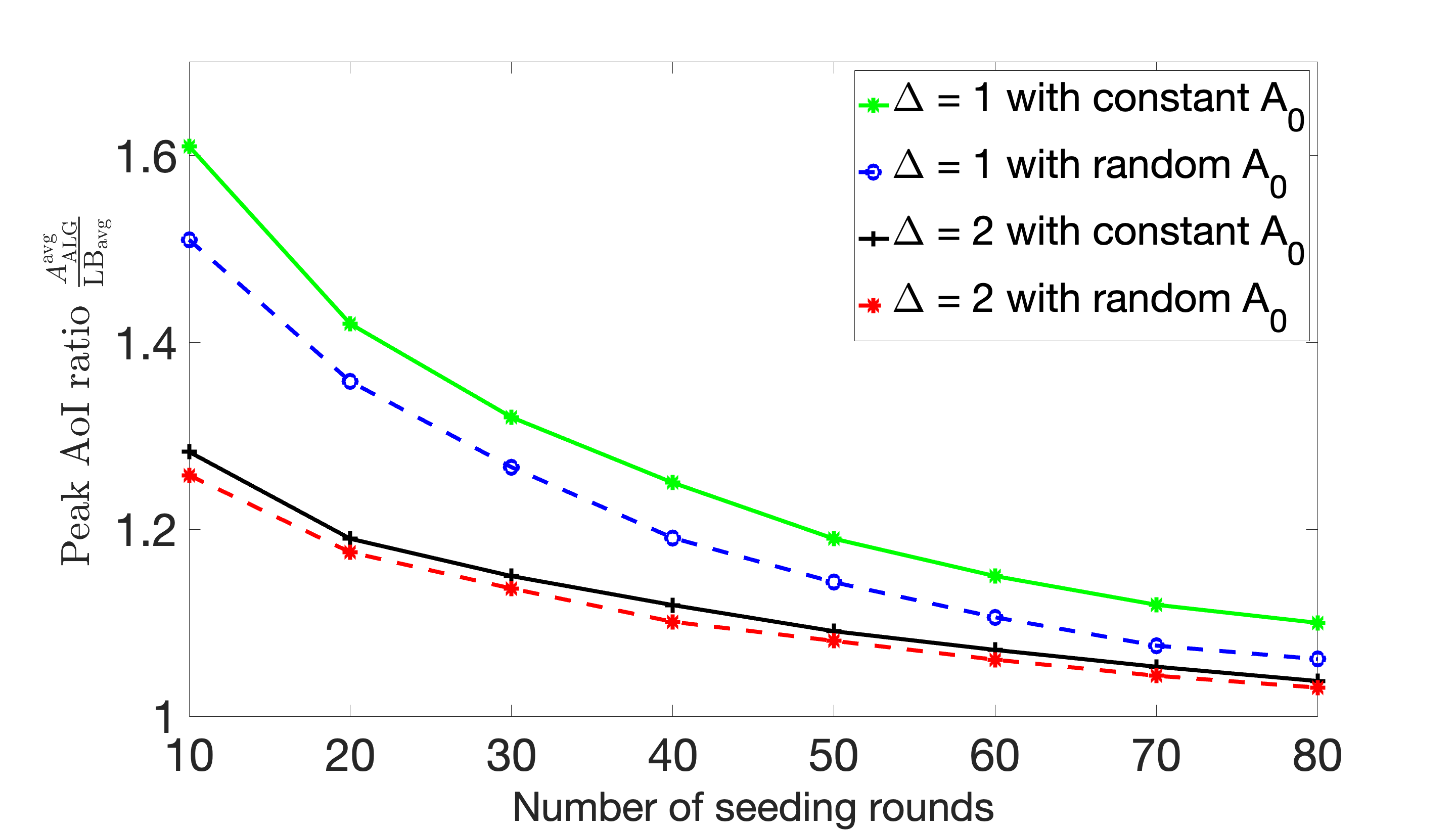}
    \caption{Average AoI performance $\frac{A^{\rm avg}_{\rm ALG}}{{\rm LB}_{\rm avg}}$.}
    \label{average_result_fig}
\end{figure}

For average AoI minimization, we not only evaluate our solutions under constant initial age $A_0$ as modeled in Section~\ref{sec_problem_statement}, but also extend to a more general setting where users' different initial ages are randomly generated with the mean equal to the constant $A_0$. We test our  Algorithm~\ref{ref_problem_01_alg} by varying the time horizon till  $T=80$ time slots and present corresponding experimental results in Fig.~\ref{average_result_fig}, by comparing with the optimum's lower bound (\ref{avg_implementation_01}) in a ratio. 
Regardless of random or deterministic $A_0$ distribution for each user, Fig.~\ref{average_result_fig} shows that our Algorithm~\ref{ref_problem_01_alg}'s empirical average AoI ratio is small and close to one, revealing Algorithm~\ref{ref_problem_01_alg}'s near-optimal performance in practical scenarios. As seed number $k$ or time horizon $T$ increases with more information updates, the performance gap between our Algorithm~\ref{ref_problem_01_alg} and the optimum (bound) further narrows, as indicated by the decreasing ratio. This also corroborates our theoretical findings in Theorem \ref{theorem_average_guarantee} as aforementioned in Section \ref{sec_average_aoi}. Moreover, Fig.~\ref{average_result_fig} also shows that
our algorithm performs even better as  $\Delta$ increases (at least in a small range), since larger $\Delta$ creates more opportunities for us to chase the optimal average AoI (i.e., $A^{\rm avg}_{OPT}$) by disseminating more nodes in each seeding interval.  Overall, Fig.~\ref{average_result_fig} has demonstrated the near-optimal performance of our Algorithm~\ref{ref_problem_01_alg} for average AoI minimization since its empirical performance ratio is only slightly larger than one (which indicates the optimum). In fact, our algorithm's real performance could even be better since the performance ratio in Fig.~\ref{average_result_fig} applies the lower bound (\ref{avg_implementation_01}) of the optimum as its denominator. 

\begin{figure}[t]
    \centering
    \includegraphics[width=9cm]{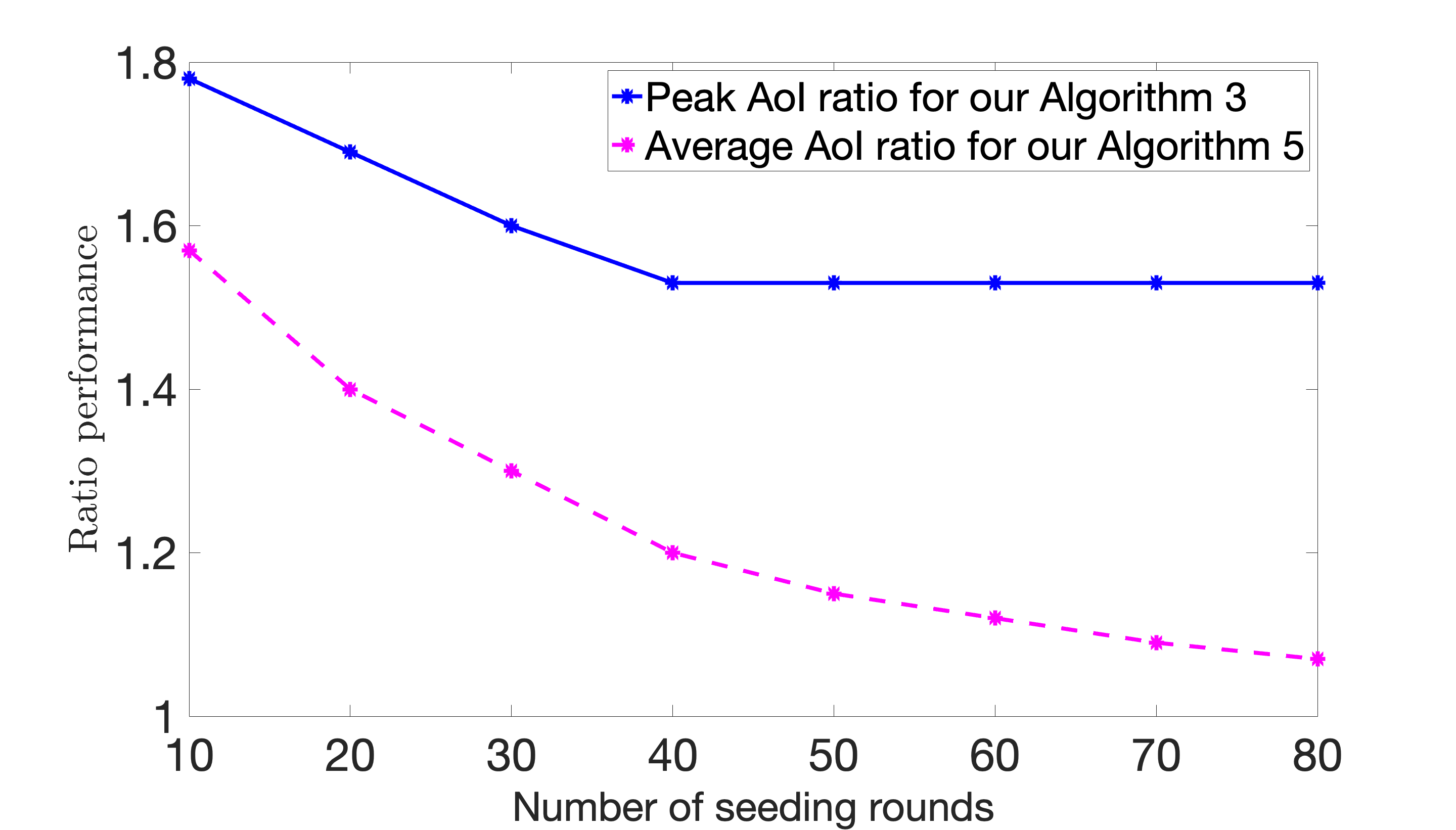}
    \caption{Peak AoI $\frac{A^{\rm peak}_{\rm ALG}}{{\rm LB}_{\rm peak}}$ under Algorithm~3 and average AoI $\frac{A^{\rm avg}_{\rm ALG}}{{\rm LB}_{\rm avg}}$ under Algorithm~5 \textit{v.s.} the total number $k$ of seeding rounds under heterogeneous propagation delays.}
    \label{multinew04}
\end{figure}
\begin{figure}[t]
    \centering
    \includegraphics[width=9cm]{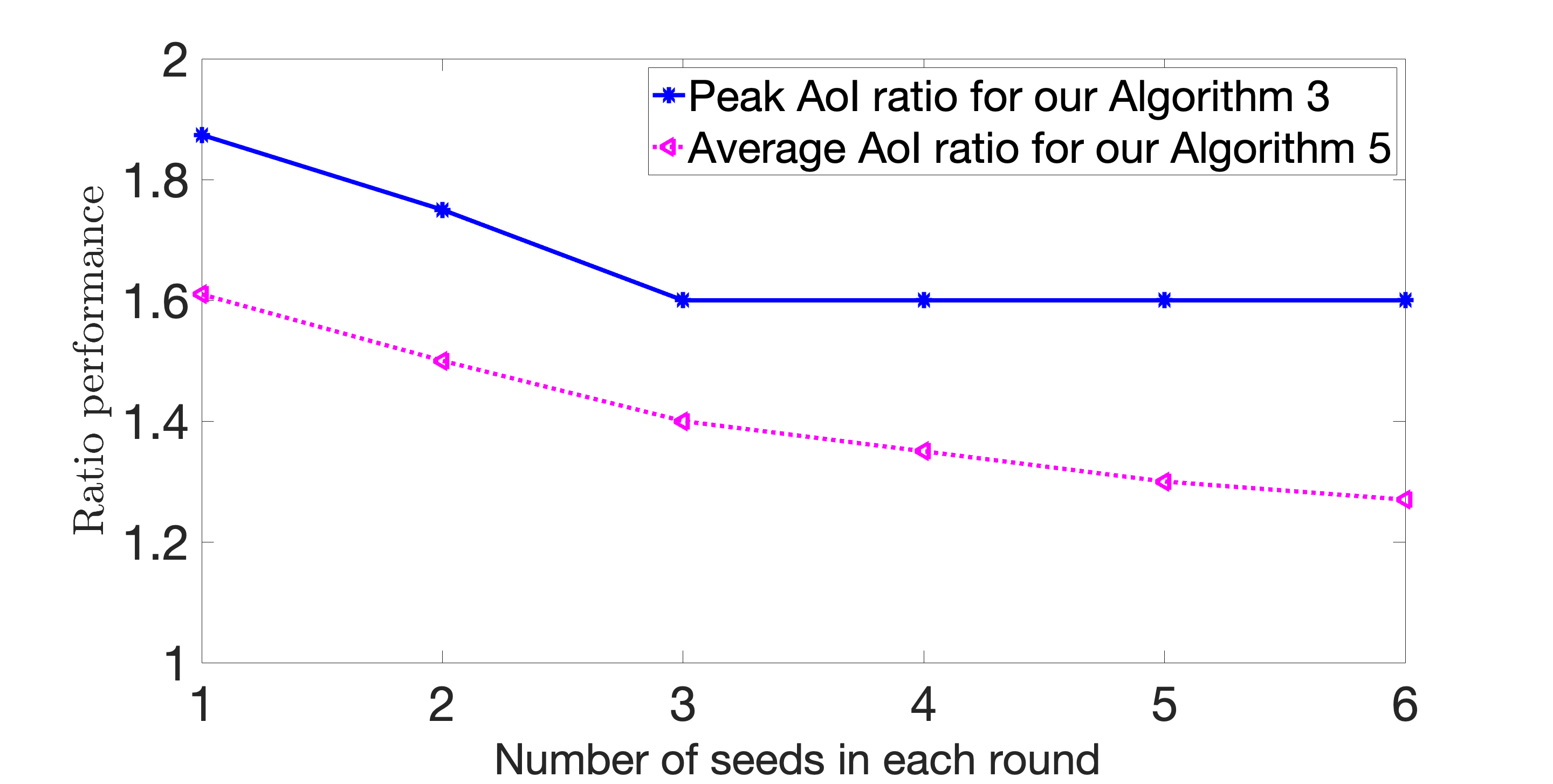}
    \caption{Average AoI $\frac{A^{\rm avg}_{\rm ALG}}{{\rm LB}_{\rm avg}}$ and peak AoI $\frac{A^{\rm peak}_{\rm ALG}}{{\rm LB}_{\rm peak}}$ \textit{v.s.} different number of per-round seeds under a total number of 10 seeding rounds.}
    \label{multinew03}
\end{figure}

Additionally, in Fig.~\ref{multinew04} and Fig.~\ref{multinew03}, we further evaluate the performances of our algorithms (developed under the uniform propagation setting) in generalized settings, including the heterogeneous propagation delays between connected nodes and the varying per-round seed budgets (i.e., the number of nodes seeded per round).

When propagation delays between connected nodes are non-uniform, the social network can be deemed weighted, with edge weights representing the varying propagation delay along each social connection.  In this implementation, edge weights for each round are randomly selected from the set $\{1, 2, 3\}$, reflecting the incomplete knowledge of our algorithms regarding the network topology.
Also, we vary the total number of seeding rounds, $k$, from 10 to 80.  Fig.~\ref{multinew04} presents our results. Compared to Figures~\ref{peak_result_fig} and~\ref{average_result_fig}, which depict the uniform delay setting, Fig.~\ref{multinew04} demonstrates that our algorithms remain robust under varying propagation delays as they continue to exhibit similarly decent ratio performances relative to the optimum.

In another generalized setting where multiple nodes can be seeded in each round of promotion updates, we vary the number of per-round seeds from 1 to 6 and consider 10 seeding rounds in total. Thanks to the round-robin approach our algorithms use for seeding from a well-designed sequence of candidates, they can readily adapt to this generalized setting. Our experimental results are presented in Fig.~\ref{multinew03}. As the number of per-round seeds increases, the total number of seeds after the fixed 10 seeding rounds increases. Similar to our earlier experimental results in Fig.~\ref{peak_result_fig}, which varies the number of seeding rounds (with one seed per round), our Fig.~\ref{multinew03} also demonstrates that both our algorithms' average and peak AoI exhibit improvements as the number of per-round seeds increases. 
Particularly, our peak AoI gradually stays at approximately 1.6 times the optimal value while our average AoI converges even more closely to the optimal value, which is around 1.3. Considering Fig.~\ref{multinew03} and our earlier Fig.~\ref{peak_result_fig} together, we further find that our algorithms perform better with a higher total number of seeds.

\section{Concluding Remarks}
To the best of our knowledge, we initiate the theoretical study of optimizing information diffusion on social networks with a multi-stage seeding process. By considering a coupling metric bridging age of information (AoI) and approximation analysis, we comprehensively study two objectives (i.e., the peak and average AoI of the network) of the problem and prove that both are NP-hard.
As a critical step, we first manage to derive closed-form expressions that trace the AoI dynamics of the network, which is highly non-trivial due to a multi-stage seeding process and the involved topology of a general social network. By focusing on a fine-tuned set of seed candidates on the diameter path, we design a fast algorithm that guarantees decent approximations as compared to the optimum for the peak AoI minimization problem. 
To minimize the average AoI, we develop a new framework that allows for algorithm design and approximation analysis, which benefits from our rigorous two-sided bound analysis on the average AoI objective.  Our framework enables us to achieve a polynomial-time algorithm that guarantees a good approximation. 
Additionally, our theoretical findings are well corroborated by extensive experiments conducted on a real social network. 

\textbf{Acknowledgement.} This work is also supported in part by the Ministry of Education, Singapore, under its Academic Research Fund Tier 2 Grant with Award no. MOE-T2EP20121-0001; in part by SUTD Kickstarter Initiative (SKI) Grant with no. SKI 2021\_04\_07; and in part by the Joint SMU-SUTD Grant with no. 22-LKCSB-SMU-053.

\nocite{langley00}
\bibliographystyle{IEEEtran}
\bibliography{mainfile}
\begin{IEEEbiography}[{\includegraphics[width=1in,height=1.23in,clip,keepaspectratio]{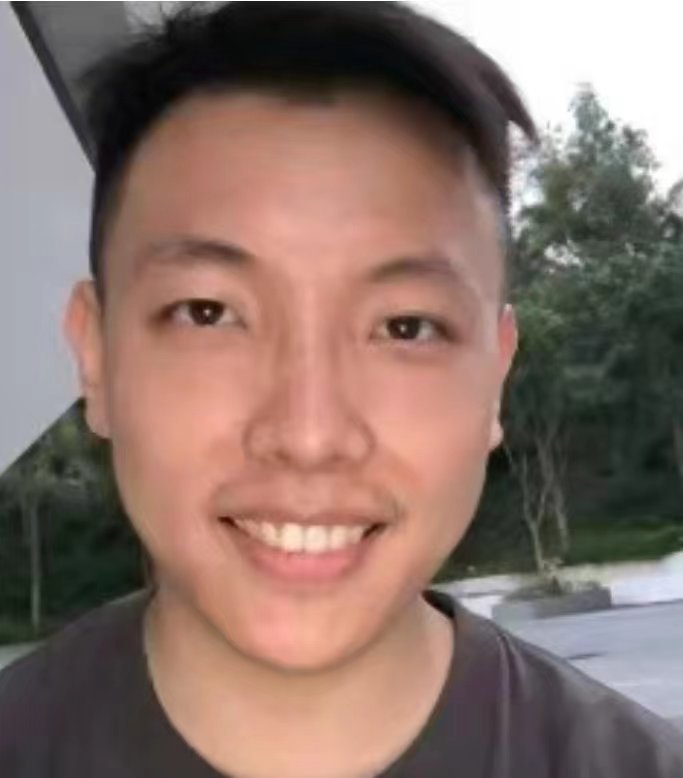}}]{Songhua Li} received
the Ph.D. degree from The City University
of Hong Kong in 2022. He is now a postdoctoral research fellow at 
the Singapore University of Technology and Design (SUTD), where he was a Visiting Scholar in 2019. His research interests include approximation and online algorithm design and analysis, combinatorial optimization, and algorithmic game theory in the context of sharing economy networks.
\end{IEEEbiography}
\begin{IEEEbiography}[{\includegraphics[width=1in,height=1.25in,clip,keepaspectratio]{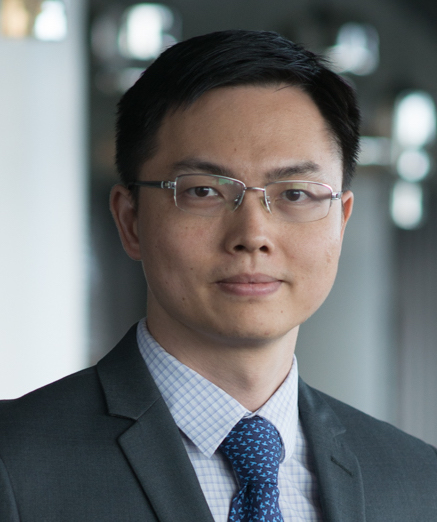}}]{Lingjie Duan} (S’09-M’12-SM’17) received the Ph.D. degree from The Chinese University of Hong Kong in 2012. He is an Associate Professor at the Singapore University of Technology and Design (SUTD) and is an Associate Head of Pillar (AHOP) of Engineering Systems and Design. In 2011, he was a Visiting Scholar at the University of California at Berkeley, Berkeley, CA, USA. His research interests include network economics and game theory, network security and privacy, energy harvesting wireless communications, and mobile crowdsourcing. He is an Associate Editor of IEEE/ACM Transactions on Networking and IEEE Transactions on Mobile Computing. He was an Editor of IEEE Transactions on Wireless Communications and IEEE Communications Surveys and Tutorials. He also served as a Guest Editor of the IEEE Journal on Selected Areas in Communications Special Issue on Human-in-the-Loop Mobile Networks, as well as IEEE Wireless Communications Magazine. He is a General Chair of WiOpt 2023 Conference and is a regular TPC member of some other top conferences (e.g., INFOCOM, MobiHoc, SECON). He received the SUTD Excellence in Research Award in 2016 and the 10th IEEE ComSoc Asia-Pacific Outstanding Young Researcher Award in 2015. 
\end{IEEEbiography}


\clearpage

\newpage
\appendix

\section{Missed Proofs in Section~\ref{sec_aoi_formulation}}\label{appendix_sec_aoi_formulation}
\subsection{Proof of Lemma~\ref{lemma_aoiof_vi}} \label{formulation_aoiof_vi_appendix}
\begin{proof}
Recall that the AoI of each node $v_i$ will increase linearly in each time slot $[t,t+1)$ for $t=\{0,..., T-1\}$, see our model of AoI evolution in Section~\ref{sec_problem_statement}. Given the set $S_x$ of dynamically selected seed nodes,  we now discuss in the following the AoI of a node $v_i$ at an integer time point $\overline{t}$, i.e., $A(v_i,\overline{t})$. Since $A(v_i,0)=A_{0}$, we only need to discuss the AoI of $v_i$ at time $\overline{t}\in\{1,...,T-1\}$. Accordingly, we discuss the following two cases. 

\noindent \textbf{Case 1.} $\Omega(v_i,\overline{t})$ is empty.

Since no seed is available to update $v_i$ at time $\overline{t}$, the AoI of $v_i$ will increase linearly based on $A(v_i,\overline{t}-1)$, i.e., 
\begin{equation}
A(v_i,\overline{t})=A(v_i,\overline{t}-1)+1.
\end{equation}

\noindent \textbf{Case 2.}  $\Omega(v_i,\overline{t})$ is not empty.

There exist seeds from set $\Omega(v_i,\overline{t})$ that update de facto the AoI of $v_i$ at time $\overline{t}$. Further, $A(v_i,\overline{t})$ will be updated by the freshest information up to time $\overline{t}$. Accordingly, we have
\begin{equation}
A(v_i,\overline{t})=\min\limits_{s_x\in\Omega}\{1+\overline{t}-t_x\}=\min\limits_{s_x\in\Omega}\{1+dist(v_i,s_x)\},
\end{equation}
where the second equation is due to the fact that $\overline{t}=t_x+dist(s_x,v_i)$. In summary, we have the AoI of $v_i$ at an arbitrary integer time point $\overline{t}$ as follows:
\begin{small}
\begin{equation}\label{evolution04}
\begin{split}
&A(v_i,\overline{t})\\
&=\left\{\begin{matrix}
 A_{0},&{\rm\;if\;} \overline{t}=0,\\   
A(v_i,\overline{t}-1)+1, &{\rm\;if\;}\overline{t}\geq 1 {\rm\;and\;}\Omega(v_i,\overline{t})=\varnothing, \\ 
\min\limits_{s_x\in\Omega(v_i,\overline{t})}\{1+dist(v_i,s_x)\}, &{\rm\;if\;} \overline{t}\geq 1 {\rm\;and\;}\Omega(v_i,\overline{t})\neq \varnothing.
\end{matrix}\right.
\end{split}
\end{equation}
\end{small}
Furthermore, at an arbitrary time $t$, $A(v_i,t)$ (i.e., the  AoI of an arbitrary node $v_i$) can be written as shown in Lemma~\ref{lemma_aoiof_vi}.
\end{proof}
\subsection{Missed Proof in Lemma~\ref{prop_dicontinuity_01}}\label{appendix_prop_dicontinuity_01}
\begin{proof}
Note that, within the time horizon $[0, T]$, $A(v_i,t)$ drops only when $v_i$ is updated its AoI by some seed.
Since each seed $s_x\in S_k$ is selected at timestamp $t_x$ and the AoI of $v_i$ takes unit time to disseminate along each edge, $v_i$ would be updated by $s_x$ at time $t_x+dist(s_x,v_i)$. This further implies that any discontinuity point of $v_i$ can be found in the following set 
\begin{equation}
    \{t_x+dist(s_x,v_i)|x\in {1,...,k}\}.
\end{equation}
This concludes the proof. 
\end{proof}

\subsection{Missed Proof in Lemma~\ref{prop_dicontinuity_02}}\label{appendix_prop_dicontinuity_02}
\begin{proof}
According to Definition~\ref{DiscontinuityPoint}, we know that a time point $t_j+dist(s_j,v_i)$ is not a discontinuity point if 
\begin{equation}  \Omega(v_i,t_j+dist(s_j,v_i))=\varnothing.
\end{equation}
Since $1\leq t_j< t_j+dist(s_j,v_i)$, we further have

\begin{equation}\label{derivation_removal_01}
\begin{split}
    &\Omega(v_i,t_j+dist(s_j,v_i))=\varnothing\\
    &\Rightarrow dist(v_i,s_j)\geq A(v_i,t_j+dist(s_j,v_i)-1)\\
    &\Rightarrow \underbrace{1+dist(v_i,s_j)}_{{\rm AoI\; of \;}s_j{\rm\;at\;}t_j+dist(s_j,v_i)}\\
    &\geq \underbrace{A(v_i,t_j+dist(s_j,v_i)-1)+1}_{{\rm AoI\;of\;some\;other\;source\;at\;}t_j+dist(s_j,v_i)}.
    \end{split}
\end{equation}
implying that there exists a seed node, denoted as $s_y$, that is selected later than $s_j$ but diffuses its information to $v_i$ no later than time $t_j+dist(s_j,v_i)$, i.e., 
\begin{equation}
    t_y>t_j  {\rm \;and\;} t_y+dist(s_y,v_i)\leq t_j+dist(s_j,v_i).
\end{equation}
This completes our proof.
\end{proof}
%

%
%
\subsection{Missed Proof in Theorem~\ref{lemma_aoi_peak_v_i}}\label{appendix_lemma_aoi_peak_v_i}
\begin{proof} 
Given $S_k$ of selected seeds, we first discuss the peak AoI of an arbitrary node $v_i\in V$. 
By Proposition~\ref{theorem_functionof_vi}, we know that the $A(v_i,t)$ achieves a supremum at the right-endpoint in each of the following time intervals, respectively.
\begin{equation}
\begin{split}
    &[0, t_{i_1}+dist(s_{i_1},v_i)),\\
    &[t_{i_1}+dist(s_{i_1},v_i),t_{i_2}+dist(s_{i_2},v_i))\\
    &,...,\\
    &[t_{i_{k_i-1}}+dist(s_{i_{k_i-1}},v_i),t_{i_{k_i}}+dist(s_{i_{k_i}},v_i)),\\
&[t_{i_{k_i}}++dist(s_{i_{k_i}},v_i),T].
\end{split}
\end{equation}
With this insight, we further get 
\begin{equation}\label{appendix_theorem1_eq01}
\begin{split}
    &A_i^{\rm peak}\\
    &=\sup\limits_{ t\in[0,T]}\{A(v_i,t)\}\\
    & = \max\{\max\limits_{j\in\{1,...,k_i\}}\{A(v_i,t_{i_j}+dist(s_{i_j},v_i))\},A(v_i,T)\}.
    \end{split}
\end{equation}
By substituting (\ref{appendix_theorem1_eq01}) back to (\ref{def_peak_aoi_network}), the $A^{\rm peak}$ over the time horizon $[0,T]$ can be expressed as follows
\begin{equation}
\begin{split}
A^{\rm peak}=\max\limits_{v_i\in V}\{A^{\rm peak}_i\}=(\ref{eq_peak_vi}),
\end{split}
\end{equation}
in which the second equation is due to Equation (\ref{eq_functionof_vi}).
\end{proof}
\subsection{Missed Proof in Theorem~\ref{avgaoi_formulation_lemma}}\label{appendix_avgaoi_formulation_lemma}
\begin{proof}
According to (\ref{def_avg_aoi_network}), we have 
\begin{equation}\label{averageaoi_network_formulation}
    A^{\rm avg}=\frac{1}{nT}\sum\limits_{v_i\in V}\int_{t=0}^T A(v_i,t)dt.
\end{equation}
Note that  $A(v_i,t)$ is a non-negative piece-wise function over the whole time horizon $[0,T]$. Given $S_k$ of dynamically selected seeds, we first look at the average AoI of a node $v_i\in V$. To this end, we partition the horizon $[0,T]$ into the following three intervals 
\begin{equation}\label{appendix_theorem2_intervals}
\begin{split}
    &[0,t_{i_1}+dist(s_{i_1},v_i)),\\
    &[t_{i_1}+dist(s_{i_1},v_i),t_{i_{k_i}}+dist(s_{i_{k_i}},v_i)),\\
    &[t_{i_{k_i}}+dist(s_{i_{k_i}},v_i),T].
\end{split}
\end{equation}
Due to Proposition~\ref{theorem_functionof_vi}, we now take
integrals of function $A(v_i,t)$ over the three intervals above in the above (\ref{appendix_theorem2_intervals}), respectively:
\begin{equation}\label{term1_oftheaverageaoi_formulation}
\begin{split}
&\int\limits_{t=0}^{t_{i_1}+dist(s_{i_1},v_i)} A(v_i,t) dt\\
&=\frac{(t_{i_1}+dist(s_{i_1},v_i))\cdot (2A_{0}+t_{i_1}+dist(s_{i_1},v_i))}{2},
\end{split}
\end{equation}

and 
\begin{equation}\label{term2_oftheaverageaoi_formulation}
\begin{split}
&\int\limits_{t=t_{i_1}+dist(s_{i_1},v_i)}^{t_{i_{k_i}}+dist(s_{i_k},v_i)}A(v_i,t)  dt\\
&=\sum\limits_{j=2}^{i_{k_i}}\frac{(2+2\cdot dist(s_{i_{j-1}},v_i)+\Lambda_{ij})\cdot \Lambda_{ij}}{2},
\end{split}
\end{equation}
where for $ \Lambda_{ij}=t_{i_j}-t_{i_{j-1}}+dist(s_{i_j},v_i)-dist(s_{i_{j-1}},v_i)$ holds for each $j\in\{2,...,k_i\}$, and
\begin{equation}\label{term3_oftheaverageaoi_formulation} 
\begin{split}
&\int_{t=t_{i_{k_i}}+dist(s_{i_{k_i}},v_i)}^{T}A(v_i,t) dt \\&=\frac{(T-t_{i_{k_i}}-dist(s_{i_{k_i}},v_i))\cdot (3+2\cdot dist(s_{i_{k_i}},v_i))}{2}.
\end{split}
\end{equation}

By applying  (\ref{term1_oftheaverageaoi_formulation}),(\ref{term2_oftheaverageaoi_formulation}), and (\ref{term3_oftheaverageaoi_formulation}) to Equation (\ref{averageaoi_network_formulation}), the average AoI of the network follows
\begin{equation}
\begin{split}
&A^{\rm avg}\\&=\frac{1}{nT}\sum_{v_i\in V} \int_{t=0}^{T}A(v_i,t)  dt\\&=\frac{1}{nT}\sum_{v_i\in V}\sum\limits_{j=1}^{k_i+1} \frac{(2A_{i{j-1}}+\Lambda_{ij})\cdot\Lambda_{ij}}{2},
    \end{split}
\end{equation}
in which  $A_{ij}$ and  $\Lambda_{ij}$ refer to Equation (\ref{formulation_Aij}) and Theorem \ref{avgaoi_formulation_lemma}, respectively. This completes the proof.
\end{proof}

%
\subsection{Missed Proof in Theorem~\ref{theorem_nphardness_avg}}\label{appendix_NP_hardness_theorem_Avg}
\begin{proof}
The NP-hardness of the average AoI minimization problem is shown by a reduction from the well-known NP-hard set cover problem. Given a ground set $U=\{u_1,...,u_q\}$ of $q$ elements and a collection $\mathcal{U}=\{U_1,...,U_p\}$ of $p$ subsets of $U$, the set cover problem aims to find $k$ ($< p$) subsets from $\mathcal{U}$ such that their union is equal to $U$. Given an instance of the set cover problem, we construct a graph of our problem by the following steps, where we set $T=3$ and $\Delta\rightarrow 0$.
\begin{itemize}
    \item \textit{Step 1.} Construct a set $V_1$ of $q$ nodes that correspond to those $q$ elements in the ground set $U$, respectively. Please refer to the top row of circles in Fig.~\ref{fig_nphardness_avg}. 
    \item \textit{Step 2.} Construct a set  $V_2$ of $p$ nodes that correspond to those $p$ subsets in $\mathcal{U}$, respectively. Please refer to those rectangles in the middle of Fig.~\ref{fig_nphardness_avg}. 
    \item \textit{Step 3.} For each  $u_x\in V_1$ and each $U_y\in V_2$, an edge $(u_x,U_y)$ is constructed if $u_x\in U_y$. Accordingly, we obtain our first edge set as $E_1$.
    \item \textit{Step 4.} Find, in the current graph $G(V_1\cup V_2,E_1)$, the maximum degree of a node of $V_1$, which is denoted as 
  \begin{equation}
      \alpha =\max\limits_{u_x\in V_1}\{deg(u_x)\}.
  \end{equation}  
     For each $U_y\in V_2$, we construct a set $\{v^y_1,...,v^y_{\alpha+1}\}$ of ($\alpha+1$) dummy nodes that correspond to $U_y$\footnote{Intuitively, these dummy nodes ensures that an optimal solution of our problem only seeds from $V_2$, which further ensures our solution to be one to the set cover problem. A formal discussion regarding this is given later in Proposition~\ref{claim_condition_ofoursolultion_nphardness}.}. Accordingly, we obtain another node set $V_3$ consisting of all the dummy nodes, i.e., $$V_3\triangleq\mathop{\cup} \limits_{j\in [p]}\{v^j_1,...,v^j_{\alpha+1}\}$$.
    \item \textit{Step 5.} For each $U_y\in V_2$ and each $v\in V_3$, we construct an edge $(U_y,v)$. Accordingly, we have another edge set $E_2$.
\end{itemize}
By the five steps above, we obtain our graph as $G(V_1\cup V_2\cup V_3,E_1\cup E_2)$.  Fig.~\ref{fig_nphardness_avg} illustrates an example of our graph construction that is converted from the set cover problem (where $U=\{u_1,u_2,....,u_5\}$, $\mathcal{U}=\{\{u_1,u_3,u_5\},\{u_2,u_4\},\{u_3,u_5\}\}$, and $k=2$).
\begin{figure}
    \centering \includegraphics[width=8cm]{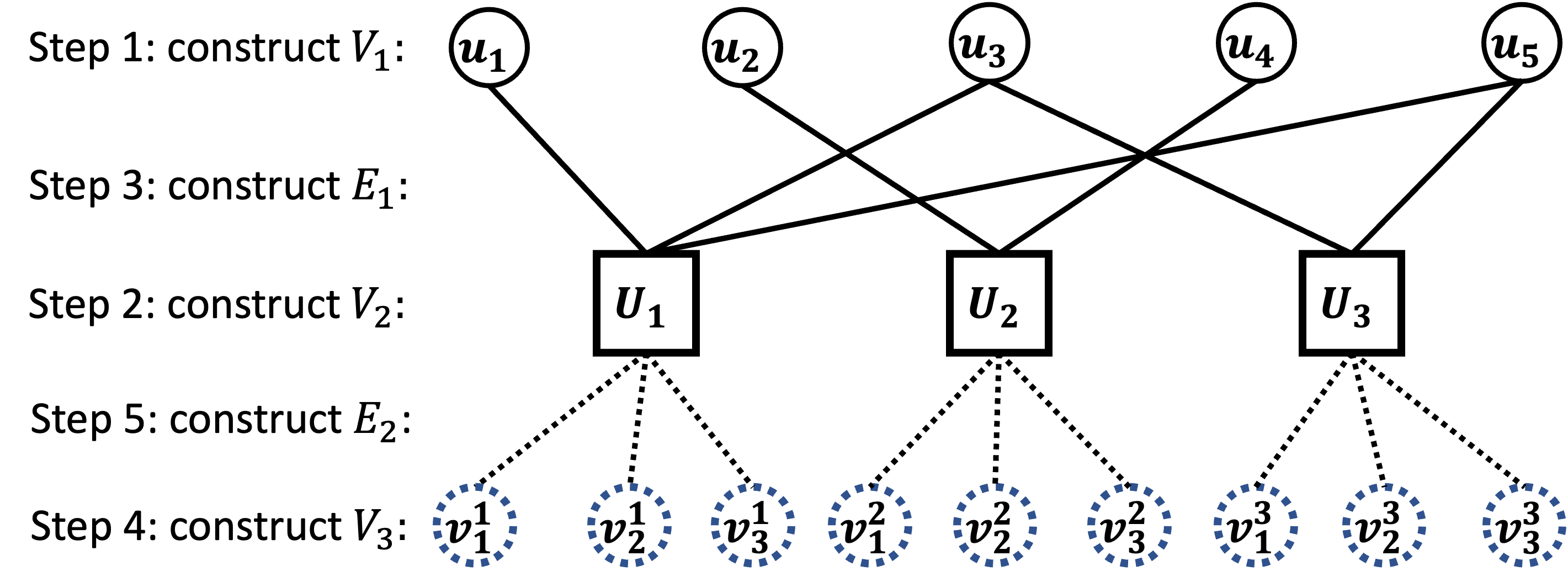}
    \caption{The graph converted from the set cover problem.}
    \label{fig_nphardness_avg}
\end{figure}

Now, we consider the average AoI of our graph $G(V_1\cup V_2\cup V_3, E_1\cup E_2)$. Since $\Delta\rightarrow 0$, the $k$ seeds could be regarded to be selected together,  implying by Lemma~\ref{prop_dicontinuity_01} that the AoI of each node in $G(V_1\cup V_2\cup V_3, E_1\cup E_2)$ drops at most once in the time horizon $[0,3)$. For ease of exposition, we partition nodes in the graph into the following types according to their average AoI (as per Theorem~\ref{lemma_aoi_peak_v_i}).
\begin{itemize}
    \item Nodes of type one refer to those seed nodes that share a common average AoI as $\frac{A_0+A_0+1}{2}+\frac{(1+3)\cdot 2}{2}=A_0+4.5$. Clearly, there are $k$ nodes of type one.
    \item Nodes of type two refer to those non-seed nodes that have at least one neighbor node selected as a seed, i.e., type-two nodes share a common average AoI as $2A_0+4.5$. Totally, there are are $|N_G(S_k)|$ nodes of type two.
    \item Nodes of type three refer to those non-seed nodes that have no neighbor node selected as a seed. That is, type-three nodes share a common average AoI as $3A_0+4.5$. Totally, the number of type three nodes is  
    \begin{equation}
        q+(\alpha+2)\cdot p-k-|N_G(S_k)|.
    \end{equation}
\end{itemize}
By Theorem~\ref{avgaoi_formulation_lemma}, we further have 
\begin{small}
\begin{equation}\label{eq_nphardness_exaple_formulate_aoi}
\begin{split}
    &A^{\rm avg}\\&=\frac{4.5 (2p+p\alpha+q) +A_0 (6p+3\alpha p+3q-2k-|N_G(S_k)|)}{3q+3p\alpha+6p}\\
    &=\frac{(9+4.5\alpha+6A_0+3\alpha A_0) p+(4.5+3A_0)q}{3q+3p\alpha+6p}\\
    &\quad-\frac{2A_0 k+A_0\cdot |N_G(S_k)|}{3q+3p\alpha+6p}\\
    &= C_1-C_2\cdot |N_G(S_k)|,
\end{split}
\end{equation}
\end{small}
in which
\begin{equation*}
\begin{split}
C_1=\frac{(9+4.5\alpha+6A_0+3\alpha A_0) p+(4.5+3A_0)q-2A_0 k}{3q+3p\alpha+6p}
    \end{split}
\end{equation*}
 and $C_2=\frac{A_0}{3q+3p\alpha+6p}>0$ are both constant. The following proposition characterizes an optimal solution of ours.
\begin{proposition}\label{claim_condition_ofoursolultion_nphardness}
In our average AoI minimization problem (that is converted from the set cover problem), an optimal solution only seeds in $V_2$. 
\end{proposition}
\begin{proof} [\textbf{Proof of Proposition \ref{claim_condition_ofoursolultion_nphardness}}]
For the sake of contradiction, suppose that there exists an optimal solution that seeds some $s_j$ which is not in $V_2$, i.e.,  $s_j\in V_1\cup V_3$. For analytical tractability, we denote $\overline{V_2}\triangleq V_2-V_2\cap S_k$ as the set of nodes in $V_2$ that are not selected as seeds. Note that $p> k$, which tells 
$\overline{V_2}\neq\varnothing$.  Denote  the set of nodes in $V_2$ that are associated with $v_j$ as 
\begin{equation}
    \widetilde{V_2}(v_j)\triangleq \{U_x\in V_2| (U_x,v_j)\in  E_1\cup E_2\}.
\end{equation}
Below, we discuss two cases. 

\noindent\textbf{Case 1.}
(There exists a node in the intersection $\widetilde{V_2}(v_j)\cap\overline{V_2}$, say $U_x$, that is associated with $v_j$ but not selected as a seed)

By seeding $U_x$ as the $j$-th seed $s_j$ instead, we know $|N_G(S_k)|$ does not decrease since the $U_x$ is the only neighbor of the old $s_j$.

\noindent\textbf{Case 2.} ($\widetilde{V_2}(v_j)\cap\overline{V_2}=\varnothing$)

 According to the pigeonhole principle, we know that there exists at least one node in $\overline{V}_2$, say $U_y$, such that none of its corresponding nodes in $\{v^y_1,...,v^y_{\alpha+1}\}$ is selected as a seed. By seeding $U_y$ as $s_j$, Further, we know that $|N_G(S_k)|$ increase by  at least one, which is because 
 \begin{equation*}
     deg(v_j)\leq \alpha< \alpha+1=deg(U_y)
 \end{equation*}
where $deg(v)$ indicates the number of edges that are incident to a vertex $v$. Hence,  according to (\ref{eq_nphardness_exaple_formulate_aoi}), we know that the average AoI of the network will increase when replacing a seed outside $V_2$ by someone in $V_2$.

This concludes the proof.
\end{proof}
\noindent 
With Proposition~\ref{claim_condition_ofoursolultion_nphardness}, Equation (\ref{eq_nphardness_exaple_formulate_aoi}) is equivalent to  
\begin{equation}\label{nphardness_average_eq01}
    A^{\rm avg}= C_1-C_2\cdot[(\alpha+1)\cdot k+|N_G(S_k)\cup V_1|]
\end{equation}
Since $C_2=C_1-C_2\cdot k\cdot (\alpha+1)$ in (\ref{nphardness_average_eq01}) is a constant, $ A^{\rm avg}$ decreases linearly with $|N_G(S_k)\cup V_1|$ (which indicates the number of nodes in $V_1$ that are associated with at least one selected node in $V_2$). Hence,  
\begin{equation}\label{eq1_nphard_average}
\begin{split}
&\arg\min\limits_{S_k\subseteq V_1\cup V_2\cup V_3, |S_k|=k}  A^{\rm avg}(S_k)\\
&=\arg\max\limits_{S_k\subseteq V_1\cup V_2\cup V_3, |S_k|=k}  |N_G(S_k)\cup V_1|,
\end{split}
\end{equation}
indicating that an optimal solution to our problem is exactly an optimal solution to the given set cover problem. 
\end{proof}

%
%
\section{Missed Proof in Section \ref{sec_peak_aoi}}
\subsection{Missed Proof in Proposition \ref{prop_def_three_values} }\label{appendix_prop_def_three_values} 
\begin{proof}
   Intuitively, Problem (\ref{obj_peak_mu})-(\ref{integerconstraint}) in Proposition \ref{prop_def_three_values} tries to minimize the earliest time (denoted as $T_1$) when all nodes on the diameter path $diam(G)$ is updated.

Suppose, \textit{w.l.o.g.}, that $T_1=1+(\underline{\mu}'-1)\Delta+\overline{\varsigma}'$, where $\underline{\mu}'$ and $\overline{\varsigma}'$ are integers and $0\leq \overline{\varsigma}'\leq \Delta-1$. This also implies that there are $\underline{\mu}'$ seeds selected before time $T_1$ to diffuse promotion information. 

We note that, by fine-tuning the allocations of seed candidates on the diameter path $diam(G)$, it could be guaranteed that different seeds (among those selected by time $T_1$) will diffuse their information to different user nodes on $diam(G)$ by time $T_1$. That is, we can manage to ensure that nodes on $diam(G)$ that are updated by different seeds by time $T_1$ are disjoint, which is discussed in Proposition \ref{prop_def_three_values} and later in this proof.

Next, we proceed with showing the feasibility of Problem (\ref{obj_peak_mu})-(\ref{integerconstraint}), followed by which we will show the optimality of our solution ($\underline{\mu},\overline{\varsigma}$) as presented in (\ref{eq01_prop2}) and (\ref{eq02_prop2}), respectively. 

\textbf{Problem Feasibility.}
  Since our algorithm only seeds on $diam(G)$, it is clear that each seed could update nodes on its both sides simultaneously. Thereby, by time $T_1$, each of the first $\underline{\mu}'$ seeds, say $s_i$, could update the following number of nodes. 
\begin{equation}\label{prop_2_eq01_temp}
    \underbrace{1}_{\rm i.e.,\; itself}+\underbrace{2[(\underline{\mu}'-i)\Delta+\overline{\varsigma}']}_{\rm i.e.,\; new\; nodes\; on\; diameter}
\end{equation}
By summing up (\ref{prop_2_eq01_temp}) over those $\underline{\mu}'$ seeds (that are selected by time $T_1$), we have the overall number of nodes updated by those $\underline{\mu}'$ seeds as follows

\begin{equation}
\begin{split}
&\sum\limits_{i=1}^{\underline{\mu}'} \bigg(1+2[(\underline{\mu}'-i)\Delta+\overline{\varsigma}']\bigg)\\&=\sum\limits_{j=1}^{\underline{\mu}'} \bigg(1+2(j-1)\Delta+2\overline{\varsigma}'\bigg)\\&=\underline{\mu}'+\Delta(\underline{\mu}'-1)\underline{\mu}'+2\overline{\varsigma}'\underline{\mu}' 
\end{split}
\end{equation}
in which the first equation is by letting $\underline{\mu}'-i=j-1$. This implies our constraint (\ref{enough_updates_constraint}), which guarantees that all nodes on $diam(G)$ are updated by time $T_1$.

\textbf{Solution Optimality}:
by relaxing constraint (\ref{integerconstraint}) to allow $\underline{\mu}'$ and $\overline{\varsigma}$ to be real number,
one can easily find an optimal solution as  
\begin{equation}\label{appendix_an_optimum_of_relaxted}
(\underline{\mu}^*,\overline{\varsigma}^*)=(\frac{\Delta-1+\sqrt{\Delta^2+2\Delta+4|diam(G)|\Delta+1}}{2\Delta},0)
\end{equation}
Further, by applying a rounding technique on (\ref{appendix_an_optimum_of_relaxted}), we get an optimal solution to the original Problem (\ref{obj_peak_mu})-(\ref{integerconstraint}) as in (\ref{eq01_prop2}) and (\ref{eq02_prop2}).

Finally, by carefully select seed candidate as  $\omega_{x}=v_{\zeta(x)}$ where sub-index of each $\omega_{x}$ follows
\begin{equation*}
    \zeta(x)=\Delta(\underline{\mu}-x+1)^2+\underline{\mu}-x+1+(2\underline{\mu}-2x+1)\overline{\varsigma}.
\end{equation*}
we could guarantee that nodes updated by different seeds by time $T_1= 1+(\underline{\mu}-1)\Delta+\overline{\sigma}$ do not overlap.  This completes the proof. 
\end{proof}
\subsection{Missed Proof in Lemma~\ref{theorem_peak_opt}}\label{appendix_theorem_peak_opt}
\begin{proof}
We will be showing that in a line-type social network, the peak AoI $A^{\rm peak}$ of any solution (including the optimal solution) follows
\begin{equation*}
    A^{\rm peak}\geq \max\{A_0+\overline{\varsigma},\underline{\xi}\Delta\}+1+\underline{\mu}\Delta,
\end{equation*}
where 
\begin{equation*}
    \xi=\frac{-(1+3\Delta)+\sqrt{((1+3\Delta))^2+4\Delta(1+2\underline{\mu}\Delta)}}{2\Delta}.
\end{equation*}

To this end, we first prove the following Lemma \ref{lemma_peak_linetype_01}.
\begin{lemma}\label{lemma_peak_linetype_01}
A line-type social network yields $A^{\rm peak}_{\rm OPT}\geq 1+A_0+\underline{\mu}\Delta+\overline{\varsigma}$.
\end{lemma}
\begin{proof} [\textbf{Proof of Lemma \ref{lemma_peak_linetype_01}}]
To prove Lemma \ref{lemma_peak_linetype_01}, it suffices to show that, in a line-type network, the time when every node is updated at least once is not earlier than ($1+\underline{\mu}\Delta+\overline{\varsigma}$). Denote $T_1$ as the earliest time when each node on a line-type social network is updated at least once. In a line-type social network, each seed could update nodes on both sides simultaneously. Before time $1+\underline{\mu}\Delta$, there are at most $\underline{\mu}$ seeds $S_{\underline{\mu}}=\{s_1,...,s_{\underline{\mu}}\}$ selected to disseminate information updates. By time $1+\underline{\mu}\Delta$, each seed $s_i\in S_{\underline{\mu}}$ could update at most $1+2\Delta\cdot(\underline{\mu}+1-i)$ nodes including $s_i$ itself. Totally, the number of nodes that seeds in $S_{\underline{\mu}}$ could update by time $1+\underline{\mu}\Delta$ can be written as
\begin{equation}\label{eq01_earliesttime_peak_si1}
\begin{split}
&\sum\limits_{i=1}^{\underline{\mu}} 1+2\Delta(\underline{\mu}+1-i)\\&=\underline{\mu}+\Delta(\underline{\mu}+1)\underline{\mu}\\
&\leq \mu+\Delta(\mu+1)\mu\\
&=n.
    \end{split}
\end{equation}
in which the last equation is due to (\ref{eq01_prop2}).

On the other hand, we know that there are $\overline{\mu}$ seeds in $S_{\overline{\mu}}=\{s_1,...,s_{\overline\mu}\}$ selected by time $1+\overline{\mu}\Delta$. Accordingly, the number of nodes that seeds in $S_{\overline{\mu}}$ update by time $1+\overline{\mu}\Delta$ follows
\begin{equation}\label{eq02_earliesttime_peak_si1}
    \begin{split}
&\sum\limits_{i=1}^{\overline{\mu}} 1+2\Delta(\overline{\mu}+1-i)\\
&=\overline{\mu}+\Delta(\overline{\mu}+1)\overline{\mu}\\
&\geq \mu+\Delta(\mu+1)\mu\\
&=n.
    \end{split}
\end{equation}
By (\ref{eq01_earliesttime_peak_si1}) and (\ref{eq02_earliesttime_peak_si1}), we know that 
\begin{equation}
T_1\in[1+\underline{\mu}\Delta,1+\overline{\mu}\Delta].
\end{equation}
To further figure out the exact $T_1$, we discuss two cases. 

\noindent
 \textbf{Case 1.} $\mu$ is an integer. 
 
Then, $\underline{\mu}=\overline{\mu}=\mu$, yielding that 
$T_1=1+\underline{\mu}\Delta$. Since 
$\mu+\Delta(\mu+1)\mu=n$, we have 
\begin{equation}
    \overline{\varsigma}=\left \lceil \frac{n-(\underline{\mu}\Delta +1)(\underline{\mu}+1)}{2(\underline{\mu}+1)} \right \rceil=0.
\end{equation}
In other words, $T_1=1+\underline{\mu}\Delta+\overline{\varsigma}$.

\noindent
\textbf{Case 2.} $\mu$ is not an integer.

We have $\sum\limits_{i=1}^{\underline{\mu}} 1+2\Delta(\underline{\mu}+1-i)<n$, implying that $T_1>1+\underline{\mu}\Delta$. This tells that there are $\underline{\mu}+1$ seeds $S_{\underline{\mu}+1}$ selected by time $T_1$. By time $1+\underline{\mu}\Delta+\underline{\varsigma}$, each seed $s_i\in S_{\underline{\mu}+1}$ could update at most $1+2[(\underline{\mu}+1-i)\Delta+\underline{\varsigma}]$ nodes, including $s_i$ itself. Then, we know, on one hand, that the number of nodes that seeds in $S_{\underline{\mu}+1}$ could update by time $1+\underline{\mu}\Delta+\underline{\varsigma}$ can be written as
\begin{equation}\label{eq03_earliesttime_peak_si1}
\begin{split}
&\sum\limits_{i=1}^{\underline{\mu}+1}1+2[(\underline{\mu}+1-i)\Delta+\underline{\varsigma}]\\
&=\underline{\mu}+1+2\underline{\varsigma}(1+\underline{\mu})+\Delta(\underline{\mu}+1)\underline{\mu}\\
&\leq \underline{\mu}+1+2\varsigma(1+\underline{\mu})+\Delta(\underline{\mu}+1)\underline{\mu}=n.
\end{split}
\end{equation}
On the other hand, by time $1+\underline{\mu}\Delta+\overline{\varsigma}$, the number of nodes that seeds in  $S_{\underline{\mu}+1}$ could update can be written as
\begin{equation}\label{eq04_earliesttime_peak_si1}
    \begin{split}
&\sum\limits_{i=1}^{\underline{\mu}+1}1+2[(\underline{\mu}+1-i)\Delta+\overline{\varsigma}]\\
&=\underline{\mu}+1+2\overline{\varsigma}(1+\underline{\mu})+\Delta(\underline{\mu}+1)\underline{\mu}\\
&\geq \underline{\mu}+1+2\varsigma(1+\underline{\mu})+\Delta(\underline{\mu}+1)\underline{\mu}=n.
\end{split}
\end{equation}
To update every node in a line-type network at least once,  (\ref{eq03_earliesttime_peak_si1}) and (\ref{eq04_earliesttime_peak_si1}) tells the earliest time is 
\begin{equation}  T_1=1+\underline{\mu}\Delta+\overline{\varsigma}.
\end{equation}
Therefore, we have
\begin{equation}
\begin{split}
A^{\rm peak}\geq \max\limits_{v_i\in V}\{ A_0+T_1\}\geq A_0+1+\underline{\mu}\Delta+\overline{\varsigma},
\end{split}
\end{equation}
in which the first inequality holds by Theorem~\ref{lemma_aoi_peak_v_i}.  Lemma  \ref{lemma_peak_linetype_01} holds readily.
\end{proof}

Further, we will be showing that $A^{\rm peak}\geq 1+\underline{\mu}\Delta+\underline{\xi}\Delta$. Denote, for ease of exposition, $num_{t}^{\geq x}$ and $num_{t}^{\leq x}$ as the number of nodes whose AoI at time $t$ are no less than and no more than $x$, respectively. Particularly, $num_{t}^{x}$ denotes the number of nodes whose AoI at time $t$ are exactly $x$. To prove Lemma \ref{theorem_peak_opt}, we have the following lemma to serve as a key ingredient.
\begin{lemma}\label{lemma_peak_linetype_02}
At time $1+\underline{\mu}\Delta$, any solution to our peak AoI minimization problem admits
\begin{itemize}
    \item $num_{1+\underline{\mu}\Delta}^{1+i\Delta}\leq 1+2i\Delta$ holds for any $i\in\{1,...,\underline{\mu}\}$,
    \item $num_{1+\underline{\mu}\Delta}^{\geq 1+(\underline{\mu}-j)\Delta}\geq (j+1)(1+2\Delta\underline{\mu}-j\Delta )$ holds for any $j\in\{0,...,\underline{\mu}-1\}$.
\end{itemize}
\end{lemma}
\begin{proof} [\textbf{Proof of Lemma \ref{lemma_peak_linetype_02}}]
To begin with, one can easily find that, at time $1+\underline{\mu}\Delta$,  a node with AoI ($1+i\Delta$) is updated by the ($\underline{\mu}+1-i$)th selected seed instead of any of the following set 
\begin{equation}
    \{S_{\underline{\mu}},...,S_{\underline{\mu}+1-i}\}.
\end{equation}
Up to time $1+\underline{\mu}\Delta$, note that the ($\underline{\mu}+1-i$)th selected seed could update at most $1+2i\Delta$ nodes, including the seed itself. This implies, for any $i\in\{1,...,\underline{\mu}\}$, that 
\begin{equation}\label{element_peak_lemmaproof_01}
num_{1+\underline{\mu}\Delta}^{1+i\Delta}\leq 1+2i\Delta.
\end{equation}
Due to Lemma~\ref{lemma_peak_linetype_01},  there are still some nodes that are not updated by time $1+\underline{\mu}\Delta$, yielding   
\begin{equation}
num_{1+\underline{\mu}\Delta}^{A_0+1+\underline{\mu}\Delta}>0,
\end{equation}
which tells that the AoI of those nodes at time $1+\underline{\mu}\Delta$ becomes $A_0+1+\underline{\mu}\Delta$.

For each $j\in\{0,...,\underline{\mu}-1\}$, we have
\begin{equation}
\begin{split}
&num_{1+\underline{\mu}\Delta}^{\geq 1+(\underline{\mu}-j)\Delta}\\
&=n-num_{1+\underline{\mu}\Delta}^{\leq 1+(\underline{\mu}-j-1)\Delta}\\
&=n-\sum\limits_{i=1}^{\underline{\mu}-j-1}num_ {1+\underline{\mu}\Delta}^{1+i\Delta}\\
&\geq n-[\underline{\mu}-j-1+\Delta(\underline{\mu}-j-1)(\underline{\mu}-j)]\\
&\geq n-[\mu-j-1+\Delta(\mu-j-1)(\mu-j]\\
&=j+1+2\Delta[(j+1)\mu-\frac{j(j+1)}{2}]\\
&\geq j+1+2\Delta[(j+1)\underline{\mu}-\frac{j(j+1)}{2}]\\
&=(j+1)(1+2\Delta\underline{\mu}-j\Delta ),
\end{split}
\end{equation}
in which the first two equations hold because, at time $1+\underline{\mu}\Delta$, the AoI of any node can be found in $\{1+\Delta,1+2\Delta,...,1+\underline{\mu}\Delta,A_0+1+\underline{\mu}\Delta\}$, the third equation holds by $\mu+(\mu^2+\mu)\Delta=n$, the first inequality holds by (\ref{element_peak_lemmaproof_01}), and the last two inequalities hold since $\mu\geq \underline{\mu}$. 
This concludes Lemma \ref{lemma_peak_linetype_02}.
\end{proof}
Now, we proceed with proving Lemma 
~\ref{theorem_peak_opt}. Denote $\Phi$ as the set of nodes whose AoI is no less than $1+\underline{\mu}\Delta$ at time $1+\underline{\mu}\Delta$. We look at the earliest time (denoted as $T_2$)  when each node is updated at least twice. Clearly, $T_2$ is strictly larger than time $1+\underline{\mu}\Delta$.  
By time $T_2$, suppose there are $y$ new seeds that are selected after time $T_1$, implying by our model that $T_2\in [1+(y+\underline{\mu})\Delta,1+(1+y+\underline{\mu})\Delta)$. In other words, the peak AoI in time period $[1+\underline{\mu}\Delta,T_2]$ follows 
\begin{equation}
    A^{\rm peak}\geq 1+(\underline{\mu}+y)\Delta.
\end{equation}
Due to Lemma~\ref{lemma_peak_linetype_02}, we have \begin{equation}\label{theorem_peak_opt_eq01}
|\Phi|=num_{1+\underline{\mu}\Delta}^{\geq 1+(\underline{\mu})\Delta}\geq 1+2\Delta\underline{\mu}.
\end{equation}
Observe, in $\Phi$, that there are at least $1+2\Delta\underline{\mu}$ nodes that are head-to-tail connected. Besides being updated by those new seeds selected in the time window $[1+\underline{\mu}\Delta,T_2]$, nodes in $\Phi$ could also be updated by other nodes with smaller AoI than ($1+\underline{\mu}\Delta$) at the time ($1+\underline{\mu}\Delta$). This implies
\begin{equation}
\begin{split}
     &2\overline{\xi}\Delta +\sum\limits_{i=1}^{\overline{\xi}} (1+2i\Delta)\\
     &= (2\Delta+1)\overline{\xi}+\Delta (\overline{\xi}+1)\overline{\xi}\\
     &\geq (2\Delta+1)\xi+\Delta (\xi+1)\xi = 1+2\underline{\mu}\Delta,
\end{split}
\end{equation}
and 
\begin{equation}
\begin{split}
     &2\underline{\xi}\Delta +\sum\limits_{i=1}^{\underline{\xi}} (1+2i\Delta) \\
     &= (2\Delta+1)\underline{\xi}+\Delta (\underline{\xi}+1)\underline{\xi}\\
     &\leq (2\Delta+1)\xi+\Delta (\xi+1)\xi = 1+2\underline{\mu}\Delta.
\end{split}
\end{equation}
Hence, at time $1+(\underline{\mu}+\underline{\xi})\Delta$, there are still some node in $\Phi$ that has not been updated since time $1+\underline{\mu}\Delta$, i.e., 
\begin{equation}
    A^{\rm peak}_{\rm OPT}\geq 1+(\underline{\mu}+\underline{\xi})\Delta.
\end{equation}
Together with Lemma~\ref{lemma_peak_linetype_01}, we have 
\begin{equation}
    A^{\rm peak}_{\rm OPT}\geq \max\{A_0+\overline{\varsigma},\underline{\xi}\Delta\}+1+\underline{\mu}\Delta.
\end{equation}
This completes proving Lemma~\ref{theorem_peak_opt}.
\end{proof}
\subsection{Missed Proof in Lemma \ref{reduce_generagraph_lemma}}\label{apppendix_reduce_generagraph_lemma}
\begin{proof}
    To prove this lemma, it suffices to show that
Algorithm~\ref{peak_cycliselection_alg} performance on any given graph $G$ is no better than the performance on its corresponding histogram-type graph $H(G)$.   In view of Theorem \ref{general_peak_approximation} and Figure \ref{histogram_social_graph}, Algorithm~\ref{peak_cycliselection_alg}  indeed reduces edges that connect those nodes outside the diameter $diam(G)$, which highly decelerates the information diffusion since Algorithm~\ref{peak_cycliselection_alg} only seeds on its diameter $diam(G)$. This concludes the proof readily. 
\end{proof}

\subsection{Missed Proof in Theorem  \ref{general_peak_approximation}}\label{appendix_general_peak_approximation}

We first prove one part of the theorem as summarized in the following lemma.
\begin{lemma}\label{lemma_histogram_feature}
Algorithm \ref{graph_reducing_alg} reduces a general graph $G(V,E)$ to a histogram-type $H(G)=(V_H,E_H)$.
\end{lemma}
\begin{proof}[Proof of Lemma \ref{lemma_histogram_feature}]
    To start with, the initialization step of Algorithm \ref{graph_reducing_alg} tells us that $diam(G)$ still remains a diameter path in the histogram-type graph $H(G)$. Further, with 
 the for-loop of Algorithm \ref{graph_reducing_alg},  it is easy to verify that every node in the given graph $G(V,E)$ is included in the node set of the histogram-type graph $H(G)$, i.e., $V_H=V$.

    Denote $v_{\vdash}$ and $v_{\dashv}$ denote the left- and the right-most nodes on the path $diam(G)$, respectively. For each node $v\in V$, denote $v_{\perp}$ as the node on $diam(G)$ closest to $v$. Then, Inequality (\ref{eq_lemma_histogram_feature}) holds readily as otherwise, it will contradict the fact that $diam(G)$ is the diameter of both $G$ and $H(G)$.
    This proves the lemma. 
\end{proof}

Then, our approximation holds steadily with Lemmas \ref{theorem_peak_opt} and~\ref{reduce_generagraph_lemma}. For formatting purposes, our remaining appendices can be found on the following page.

\section{Missed Proofs in Section~\ref{sec_average_aoi}}
\subsection{Missed Proof in Proposition~\ref{corollary_average_Aoi_over_the network}}\label{appendix_proof_for_reformulated_avgaoi}
\begin{proof}
To prove Proposition~\ref{corollary_average_Aoi_over_the network}, we first give the following lemma
\begin{lemma}\label{corollary_average_aoi}
Given $S_k$ of selected seeds, we have 
\begin{equation}\label{corollary_average_aoi_eq}
\begin{split}
       A^{\rm avg}_i=&\eta+\frac{A_0}{T}[\Delta i_1+dist(s_{i_1},v_i)]\\
       &+\frac{1}{T}\sum\limits_{j=1}^{k_i}[\Delta \cdot dist(s_{i_j},v_i)\cdot (i_j-i_{j-1})]\\
       &+\frac{1}{T}\sum\limits_{j=1}^{k_i}[0.5\Delta^2\cdot (i_j-i_{j-1})^2]
       \end{split}
\end{equation}
where $U_i=\{t_{i_j}+dist(s_{i_j},v_i)|i_j\in\{i_1,...,i_{k_i}\}\}$ can be obtained by Algorithm~\ref{DiscontinuityPoint}, $i_0=1$, and
%
\begin{equation}\label{eq_constant_averageaoi}
    \begin{split}
    \eta =& \frac{1}{2T}[2A_0-2A_0\Delta +\Delta^2-2\Delta+2k\Delta-2k\Delta^2+k^2\Delta^2]\\
&+\frac{T}{2}+(1-k)\Delta
    \end{split}
\end{equation}
\end{lemma}

In the following, we first prove Lemma~\ref{corollary_average_aoi}, after which we will show that Proposition~\ref{corollary_average_Aoi_over_the network} holds readily.
We first consider the $\sum\limits_{j=2}^{k_i} (2A_{i{j-1}}+\Lambda_{ij})\cdot \Lambda_{ij}$ as analyzed in (\ref{large_01eq}). Then, we can obtain $\sum\limits_{j=1}^{k_i+1} (2A_{i{j-1}}+\Lambda_{ij})\cdot \Lambda_{ij}$ as analyzed in (\ref{corollary_aoi_formulation_avg_eq01}), which further implies $A_i^{\rm avg}$ as discussed in (\ref{large_03eq}).
Due to formatting reasons, (\ref{corollary_aoi_formulation_avg_eq01}) and (\ref{large_03eq}) are given on the following page. This concludes the proof of Lemma~\ref{corollary_average_aoi}. In the following, we continue to prove Proposition~\ref{corollary_average_Aoi_over_the network}.

With Lemma~\ref{corollary_average_aoi} at hand, we have Proposition~\ref{corollary_average_Aoi_over_the network} by summing up (\ref{corollary_average_aoi_eq}) over user nodes in $V$, 
which inspires the social commerce network to select its best $\Delta$ as discussed in the following remark.
\begin{remark}
Note that $\eta$ is independent of $S_k$ and $v_i$. Accordingly, $2T\cdot\eta$ is a constant AoI among each node $v_i\in V$, which is also independent of the selection of $S_k$. Due to Equation (\ref{eq_constant_averageaoi}), we have 
\begin{equation*}
    \eta\propto (1-k)^2\Delta^2+2\Delta (T-A_0-1+k-Tk)
\end{equation*}
Thereby, we have $\Delta =\frac{A_0}{(k-1)^2}+\frac{T-1}{(k-1)}$ minimizes the common area of $\eta$, given $A_0$, $k$ and $T$.  
\end{remark}

\begin{small}
\begin{table*}
\begin{equation}\label{large_01eq}
    \begin{split}
      &\sum\limits_{j=2}^{k_i} (2A_{i{j-1}}+\Lambda_{ij})\cdot \Lambda_{ij}\\
&=\underbrace{[2+(i_2-i_1)\Delta+dist(s_{i_2},v_i)+dist(s_{i_1},v_i)]\cdot [(i_2-i_1)\Delta+dist(s_{i_2},v_i)-dist(s_{i_1},v_i)]}_{{\rm i.e.,\;}(2A_{i{1}}+\Lambda_{i2})\cdot \Lambda_{i2}}\\
&\;\;\;+\underbrace{[2+(i_3-i_2)\Delta+dist(s_{i_3},v_i)+dist(s_{i_2},v_i)]\cdot [(i_3-i_2)\Delta+dist(s_{i_3},v_i)-dist(s_{i_2},v_i)]}_{{\rm i.e.,\;}(2A_{i{2}}+\Lambda_{i3})\cdot \Lambda_{i3}}+\\
&\;\;\;\;\vdots\\
&\;\;\;+\underbrace{[2+(i_{k_i}-i_{k_i-1})\Delta+dist(s_{i_{k_i}},v_i)+dist(s_{i_{k_i-1}},v_i)]\cdot [(i_{k_i}-i_{k_i-1})\Delta+dist(s_{i_{k_i}},v_i)-dist(s_{i_{k_i-1}},v_i)]}_{{\rm i.e.,\;}(2A_{i{k_i-1}}+\Lambda_{ik_i})\cdot \Lambda_{ik_i}}\\
&=\underbrace{2[(i_2-i_1)\Delta+dist(s_{i_2},v_i)-dist(s_{i_1},v_i)]}_{\rm denoted\;as \;P_{21}}+\underbrace{[(i_2-i_1)^2\Delta^2+2(i_2-i_1)dist(s_{i_2},v_i)\Delta]}_{\rm denoted \;as \;P_{22}}+\underbrace{dist^2(s_{i_2},v_i)-dist^2(s_{i_1},v_i)}_{\rm denoted\;as \;P_{23}}\\
&\;\;\;
+\underbrace{2[(i_3-i_2)\Delta+dist(s_{i_3},v_i)-dist(s_{i_2},v_i)]}_{\rm denoted\;as \;P_{31}}+\underbrace{[(i_3-i_2)^2\Delta^2+2(i_3-i_2)dist(s_{i_3},v_i)\Delta]}_{\rm denoted \;as \;P_{32}}+\underbrace{dist^2(s_{i_3},v_i)-dist^2(s_{i_2},v_i)}_{\rm denoted\;as \;P_{33}}\\
&\;\;\;\;\vdots\\
&\;\;\;+\underbrace{2[(i_{k_i}-i_{k_i-1})\Delta+dist(s_{i_{k_i}},v_i)-dist(s_{i_{k_i}-1},v_i)]}_{\rm denoted\;as \;P_{k_i1}}+\underbrace{[(i_{k_i}-i_{k_i-1})^2\Delta^2+2(i_{k_i}-i_{k_i-1})dist(s_{i_{k_i}},v_i)\Delta]}_{\rm denoted \;as \;P_{k_i2}}+\underbrace{dist^2(s_{i_{k_i}},v_i)-dist^2(s_{i_{k_i-1}},v_i)}_{\rm denoted\;as \;P_{k_i3}}\\
&=\underbrace{2[(i_{k_i}-i_1)\Delta+dist(s_{i_{k_i}},v_i)-dist(s_{i_1},v_i)]}_{=\sum_{j=2}^{k_i} P_{j1}}+\underbrace{\sum\limits_{j=2}^{k_i}[(i_j-i_{j-1})^2\Delta^2+2\Delta(i_j-i_{j-1})dist(s_{i_j},v_i)]}_{=\sum_{j=2}^{k_i} P_{j2}}+\underbrace{dist^2(s_{i_{k_i}},v_i)-dist^2(s_{i_1},v_i)}_{=\sum_{j=2}^{k_i} P_{j3}}.
    \end{split}
\end{equation}
\end{table*}
\begin{table*}
\begin{equation}\label{corollary_aoi_formulation_avg_eq01}
    \begin{split}
      &\sum\limits_{j=1}^{k_i+1} (2A_{i{j-1}}+\Lambda_{ij})\cdot \Lambda_{ij}\\
       &= \underbrace{[2A_0+1+(i_1-1)\Delta+dist(s_{i_1},v_i)][1+(i_1-1)\Delta+dist(s_{i_1},v_i)]}_{{\rm i.e.,\;}(2A_{i0}+\Lambda_{i1})\cdot \Lambda_{i1}}\\
&\;\;\;+\underbrace{2[(i_{k_i}-i_1)\Delta+dist(s_{i_{k_i}},v_i)-dist(s_{i_1},v_i)]+\sum\limits_{j=2}^{k_i}[(i_j-i_{j-1})^2\Delta^2+2\Delta(i_j-i_{j-1})dist(s_{i_j},v_i)]+dist^2(s_{i_{k_i}},v_i)-dist^2(s_{i_1},v_i)}_{{\rm i.e.,\;}\sum_{j=2}^{k_i} (2A_{ij-1}+\Lambda_{ij})\cdot \Lambda_{ij}}\\
&\;\;\;+\underbrace{[T+1-(i_{k_i}-1)\Delta+dist(s_{i_{k_i}},v_i)]\cdot [T-1-(i_{k_i}-1)\Delta-dist(s_{i_{k_i}},v_i)]}_{{\rm i.e.,\;}(2A_{ik_i}+\Lambda_{ik_i+1})\cdot \Lambda_{ik_i+1}}\\
&=2A_0+2A_0\Delta i_1-2A_0\Delta +\underbrace{2A_0 dist(s_{i_1},v_i)}_{\rm denoted\;as\;Q_{11}}+[1+(i_1-1)\Delta ]^2+\underbrace{2(1+i_1\Delta-\Delta) dist(s_{i_1},v_i)}_{\rm denoted \;as\;Q_{12}}+\underbrace{dist^2(s_{i_1},v_i)}_{\rm denoted\;as\;Q_{21}}\\
&\;\;\;+2(i_{k_i}-i_1)\Delta+\underbrace{2[dist(s_{i_{k_i}},v_i)-dist(s_{i_1},v_i)]}_{\rm denoted\;as\;Q_{13}}+\sum\limits_{j=2}^{k_i}(i_j-i_{j-1})^2\Delta^2+\underbrace{\sum\limits_{j=2}^{k_i}2\Delta(i_j-i_{j-1})dist(s_{i_j},v_i)}_{\rm denoted \;as\;Q_{14}}+\underbrace{dist^2(s_{i_{k_i}},v_i)-dist^2(s_{i_1},v_i)}_{\rm denoted\;as\;Q_{22}}\\
&\;\;\;+T^2-2T(i_{k_i}-1)\Delta +(i_{k_i}-1)^2\Delta^2-1\underbrace{-2dist(s_{i_{k_i}},v_i)}_{\rm denoted\;as\;Q_{15}}\underbrace{-dist^2(s_{i_{k_i}},v_i)}_{\rm denoted \;as\;Q_{23}}\\
&=\underbrace{2A_0\cdot dist(s_{i_1},v_i)+\sum\limits_{j=1}^{k_i}2\Delta(i_j-i_{j-1})\cdot dist(s_{i_j},v_i)}_{\rm i.e.,\;\sum_{x=1}^5 Q_{1x}}+\underbrace{\sum\limits_{j=1}^3Q_{2j}}_{\rm equal\;to\;0}+
\underbrace{2A_0-2\Delta A_0+T^2+2\Delta T+\Delta^2-2\Delta}_{\rm i.e.,\;a\;constant\;}\\
&\;\;\;+2A_0\Delta i_1+2(\Delta-T\Delta-\Delta^2)i_{k_i}+\Delta^2 i_{k_i}^2+\Delta^2\sum\limits_{j=1}^{k_i}(i_j-i_{j-1})^2
    \end{split}
\end{equation}
\end{table*}

\begin{table*}
\begin{equation}\label{large_03eq}
    \begin{split}
A^{\rm avg}_i
&=\frac{1}{2T}\sum\limits_{j=1}^{k_i+1} (2A_{i{j-1}}+\Lambda_{ij})\cdot \Lambda_{ij}\\
&=\underbrace{\frac{1}{2T}[2A_0-2\Delta A_0+T^2+2\Delta T+\Delta^2-2\Delta+2(\Delta-T\Delta-\Delta^2)k+\Delta^2 k^2]}_{\rm i.e.,\;a\;constant}\\
&\;\;\;+\frac{1}{2T}\left\{2A_0\cdot [\Delta i_1+dist(s_{i_1},v_i)]+2\Delta\cdot\sum\limits_{j=1}^{k_i}[(i_j-i_{j-1})\cdot dist(s_{i_j},v_i)]+\Delta^2\sum\limits_{j=1}^{k_i}(i_j-i_{j-1})^2\right\}.
    \end{split}
\end{equation}
\end{table*}
\end{small}

\end{proof}

\end{document}